\documentclass[journal,twoside,web]{ieeecolor}

\usepackage{generic}
\usepackage{textcomp}
\usepackage{cite}
\usepackage{amsmath,amssymb,amsfonts}
\usepackage[mathscr]{eucal}
\usepackage{algorithm}
\usepackage{algorithmic}
\usepackage{inputenc}
\usepackage{graphicx}
\usepackage{amsmath,tikz}
\usetikzlibrary{matrix}
\usepackage{verbatim}
\usepackage[hidelinks]{hyperref}
\usepackage{textcomp}
\usepackage{subcaption}
\usepackage{lipsum}
\usepackage{babel}
\usepackage{blindtext}
\usepackage{cases}
\usepackage{epstopdf}
\usepackage{multirow}
\usepackage{array}

\newtheorem{theorem}{Theorem}

\newtheorem{propo}{Proposition}
\newtheorem{lemma}{Lemma}
\newtheorem{corollary}{Corollary}

\newtheorem{assumption}{Assumption}

\def\thickhline{\noalign{\hrule height1pt}}
\newcolumntype{?}{!{\vrule width 1pt}}

\floatname{algorithm}{Algorithm}
\DeclareMathOperator*{\argmin}{arg\,min}

\DeclareMathOperator{\tr}{tr}

\newcommand{\R}{\mathbb{R}}
\newcommand{\Sbb}{\mathbb{S}}

\newcommand{\Xb}{\mathbf{X}}
\newcommand{\Yb}{\mathbf{Y}}
\newcommand{\Ub}{\mathbf{U}}

\newcommand{\Mc}{\mathcal{M}}
\newcommand{\Sc}{\mathcal{S}}
\newcommand{\Fc}{\mathcal{F}}
\newcommand{\Cc}{\mathcal{C}}
\newcommand{\Uc}{\mathcal{U}}
\newcommand{\Nc}{\mathcal{N}}
\newcommand{\Oc}{\mathcal{O}}

\newcommand{\LMI}{\text{LMI}}

\newcommand{\tp}{^\top}
\newcommand{\beq}{\begin{equation}}
	\newcommand{\eeq}{\end{equation}}
\newcommand{\bmat}{\begin{bmatrix}}
	\newcommand{\emat}{\end{bmatrix}}

\newcommand\m[1]{\begin{bmatrix}#1\end{bmatrix}} 

\definecolor{purple}{rgb}{0.62, 0.0, 0.77}

\def\BibTeX{{\rm B\kern-.05em{\sc i\kern-.025em b}\kern-.08em
    T\kern-.1667em\lower.7ex\hbox{E}\kern-.125emX}}
\begin{document}

\title{Data-driven optimal control via  linear programming:  boundedness guarantees}

\author{Lucia Falconi, Andrea Martinelli, and John Lygeros 
\thanks{This work was supported in part by Fondazione CARIPARO (Borse di Dottorato CARIPARO 2020), and by the European Research Council under the Horizon 2020 Advanced Grant 787845 (OCAL). \emph{(L. Falconi and A. Martinelli contributed equally to this work.) (Corresponding author: L. Falconi)}
}
\thanks{L. Falconi is with the Department of Information Engineering, University of Padova, Via Gradenigo 6/B, 35131 Padova, Italy (email: lucia.falconi@phd.unipd.it).
}
\thanks{A. Martinelli and J. Lygeros are with the Automatic Control Laboratory, ETH Zurich, Physikstrasse 3, 8092 Zurich, Switzerland (email: andremar@ethz.ch; jlygeros@ethz.ch).
}
}

\maketitle

\begin{abstract} 

The linear programming (LP) approach is, together with value iteration and policy iteration, one of the three fundamental methods to solve optimal control problems in a dynamic programming setting. Despite its simple formulation, versatility, and predisposition to be employed in model-free settings, the LP approach has not enjoyed the same popularity as the other methods. The reason is the often poor scalability of the exact LP approach and the difficulty to obtain bounded solutions for a reasonable amount of constraints. We mitigate these issues here, by investigating fundamental geometric features of the LP and developing sufficient conditions to guarantee finite solutions with minimal constraints. In the model-free context, we show that boundedness can be guaranteed by a suitable choice of dataset and objective function.
\end{abstract}

\begin{IEEEkeywords}
Approximate dynamic programming, data-driven control, linear programming, optimal control    \end{IEEEkeywords}

\section{Introduction} \label{sec:Introduction}

\subsection{Literature review}

Optimal control problems arise in most engineering and scientific disciplines, from robotics to bioengineering to finance \cite{BoborowRobotics,KuoBioengineering,BERTSIMASFinance}, to name but a few. One approach to solve optimal control problems was developed in the 1950's by Bellman \cite{BellmanDP}. It is based on the so-called \textit{principle of optimality} to derive a functional equation in the state and input variables, the Bellman equation, characterizing the optimal feedback controller. The methods to solve the Bellman equation are collectively known as \textit{dynamic programming} (DP), and typically rely on three fundamental techniques: value iteration, policy iteration, and linear programming (LP) \cite{bertsekas_volI,bertsekas_volII}. The foundations of all three techniques lie in the monotonicity and contractivity properties of the Bellman operator, implicitly defined by the Bellman equation \cite{bertsekas2018abstract}. The main difficulty affecting DP methods is the rapid growth of the computation needed to solve the Bellman equation with the number of variables involved, known as the \textit{curse of dimensionality}. Throughout the years, several approximation techniques have been developed to mitigate this problem, collectively known as \textit{approximate} dynamic programming (ADP) \cite{bertsekas_volII,bertsekas2019RL,PowellADP}. 

In this context, the present manuscript focuses on the LP approach to ADP. Introduced by Manne in the 1960's \cite{ManneLP}, the method is based on the idea that the solution to the Bellman equation can be cast as the optimizer of an LP. In the mid-1980's, Schweitzer presented an approximate version of the LPs for finite state-action spaces, by restricting the function space to a finite family of linearly parameterized functions \cite{SCHWEITZER1985568}. Nevertheless, it is only in the early 2000' with the work by De Farias and Van Roy \cite{de2003linear,deFarias2004} that error bounds due to function approximation and constraint sampling are developed. 

When dealing with infinite state and action spaces, the exact infinite dimensional LPs have to be approximated by
finite dimensional ones \cite{hernandez1998LP,LasserreDTMCP,KlabjanLP,John2017}. Several methods that build upon these approximation schemes have been discussed in recent years, for instance by using Sum-of-Squares polynomials \cite{summers2013sos} or iterated Bellman inequalities \cite{wang2015approximate}.

In many real-world applications one has to operate a system without the knowledge of its dynamics. Learning how to optimally regulate a system is possible by letting it interact with the environment and collecting information about the resulting state transitions and costs (or rewards) over time, an approach classically known as \textit{reinforcement learning} (RL) \cite{Sutton1998,BUSONIU2018}. To extend RL methods in the LP approach framework, one can reformulate the Bellman equation in terms of the $Q$-function \cite{Watkins1992}, and set up a new class of LPs based on the Bellman operator for $Q$-functions. The main advantage of the $Q$-function formulation is that one can extract a policy directly from $Q$, without knowledge of the system dynamics or stage cost. 

Cogill and co-authors were the first to introduce an LP formulation based on $Q$-functions in 2006 \cite{cogill2006}. A first analysis on the associated error bounds can be found in \cite{SutterLP} for finite state-action spaces, and in \cite{Beuchat2016PerformanceGF} for continuous ones. Learning the optimal $Q$-function for deterministic systems by mixing LP with approximate value iteration and policy iteration schemes is investigated in \cite{banj2019} and \cite{TANZANAKIS20207}. The relaxed Bellman operator introduced in \cite{MARTINELLIAut} allows one to design efficient LPs based on $Q$-functions even for stochastic systems. In case of linear or affine dynamics, one can generate new constraints offline starting from a small dataset \cite{MartinelliCDC}, and approximate the expectation in the LP constraints with the estimators introduced in \cite{MartinelliAffine}. A taxonomy of this literature is provided in Table \ref{tab:tax}.

\subsection{Motivations}

The exact LP formulation for finite state-action spaces was viewed as primarily a theoretical result for many years, since it requires formulating an LP where the number of constraints is equal to the number of state-action pairs \cite[Ch. 3.8]{PowellADP}. The method has seen a resurgence in recent years, concurrently with major advances in LP solution methods \cite{boyd}, and especially via approximation schemes to deal with continuous state-action spaces and model-free approaches \cite{bertsekas2019RL}. Indeed, most of the above-mentioned literature has shown encouraging results, particularly for low-dimensional systems. To finally become competitive even for higher-dimensional problems, we argue that the main obstacle is the difficulty of systematically obtaining bounded solution. 

As first mentioned in \cite{MARTINELLIAut}, the number of constraints needed to keep the LP bounded grows dramatically with the dimension of the system. Indeed, we note that most of the (even recent) literature on the topic only propose control examples with minimal dimension and/or massive dataset \cite{de2003linear,deFarias2004,wang2015approximate,SutterLP,Beuchat2016PerformanceGF,banj2019,TANZANAKIS20207}. Some papers deal with this problem by introducing a regularizer to bound a certain norm of the optimal value/$Q$-function \cite{SutterLP}. This, however, can introduce a significant and hard to quantify distortion of the optimizers, since in case of unboundedness the solution is effectively determined by the regularizer and not the constraints. To the best of our knowledge, the only work to explicitly consider sufficient conditions to obtain bounded solutions without regularizers is \cite{CDC22_meyn}. They consider the ``optimal" LP formulation, that is, the LP whose solution is the optimal value/$Q$-function, and provide a sufficient condition based on the rank of a certain data-based covariance matrix. We will review this condition in Section \ref{sec:appLP_opt} together with additional insights. 

\subsection{Outline and Contributions}

The main goal of this manuscript is to provide systematic and practical conditions on the constraints to guarantee finite solution of the LPs ``associated to a policy", that is, the LPs whose solution is the value/$Q$-function associated to a fixed control policy. To do so, we first investigate some fundamental geometric features of the LPs feasible regions. 

In Section~\ref{sec:exactLP_pi} we focus on the exact LP formulations, showing that (i) the feasible region of the LPs associated to a policy is a convex cone open-ended in the ``south-west" direction, and (ii) the feasible region of the optimal LPs is a convex set that can be either bounded or unbounded.
Section \ref{sec:appLP_pi} is devoted to approximate LPs, where the function space is restricted  to a finite dimensional linear subspace
and the constraints are either sampled or observed from system exploration. We prove that the approximate LPs enjoy similar geometric features  to
their exact counterparts. Then, we use Farkas' Lemma to prove that the LPs associated to a policy are bounded if and only if the objective function lies in the conical hull of the rows of a certain data matrix. Finally, under mild assumptions, we show that if the solution is bounded, then it is optimal. 
In Section \ref{sec:smartsampl} we focus on linear quadratic problems. We show that the LP feasible region is bounded if and only if the system is reachable. Then, we specialize the previous results and discuss two systematic methods to guarantee bounded (hence optimal) solutions by properly selecting the gradient of the objective function of the LP based on the observed data. 
The numerical results in Section~\ref{sec:numerics} provide evidence for the difficulty of obtaining bounded solutions for the approximate LPs when random sampling is employed. On the other hand, they illustrate the effectiveness of the strategies in Section \ref{sec:smartsampl} where we consistently obtain bounded solutions with minimal constraints.

\begin{table*}[h]
  \centering
  \begin{tabular}{ll?ll}
    State-Action Space & Model & Value Function & $Q$-function \\ 
    \thickhline
    \multirow{2}{*}{Finite} & Deterministic & -- & -- \\
    \cline{2-4}
    & Stochastic & Manne \cite{ManneLP}, Schweitzer \cite{SCHWEITZER1985568}, De Farias \cite{de2003linear,deFarias2004}, Desai \cite{DesaiSmoothed} & Cogill \cite{cogill2006}, Sutter \cite{SutterLP} \\
    \hline
    \multirow{2}{*}{Continuous} & Deterministic & Klabjan \cite{KlabjanLP} & Banjac \cite{banj2019}, Tanzanakis \cite{TANZANAKIS20207}, Lu \cite{CDC22_meyn} \\
    \cline{2-4}
    & Stochastic & Hernández-Lerma \cite{hernandez1998LP}, Summers \cite{summers2013sos}, Wang \cite{wang2015approximate}, Esfahani \cite{John2017} & Beuchat \cite{Beuchat2016PerformanceGF}, Martinelli \cite{MartinelliCDC,MartinelliAffine,MARTINELLIAut} \\
  \end{tabular}
  \caption{Taxonomy of papers on the LP approach to ADP. The present paper deals with continuous state-action spaces, deterministic models, and $Q$-function formulation.}
  \label{tab:tax}
\end{table*}

\subsection{Notation}
$\R_{+}$ ($\R_{++}$) is the space of nonnegative (positive) real num-bers;  $\Sbb^n$ is the space of real symmetric matrices of size $n$; $I$ is the identity matrix. We denote positive semidefinite (definite) matrices $M \in \Sbb^n$ with $M \succeq 0 \; (\succ 0)$.  
Given a finite-dimensional vector-space $\Yb$, we denote by $\Sc(\Yb)  $ the vector space of  real-valued measurable functions with a finite weighted sup-norm \cite[ch.2.1]{bertsekas2018abstract}:
$$ \Sc(\Yb) = \left\{ g: \Yb \to \R \; : \; \Vert g(y) \Vert_{\infty,z} := \sup_{y \in \Yb} \frac{|g(y)|}{z(y)} < \infty  \right\}, $$
where $z: \Yb \to \R_{++}$ is a weight function.
The space of finite signed measure is defined as:
$$ \Mc(\Yb) = \left\{ c: \Yb \to \R \; : \; \int_{\Yb} |c(y)|z(y) < \infty  \right\}. $$
Finally, $ \Mc^+(\Yb)$ is the cone of non-negative measures:
$$ \Mc^+(\Yb) = \left\{ c \in \Mc(\Yb) \; : \; c(B)\geq 0,  \forall B \subseteq \Yb \text{ measurable} \right\}. $$
Given a square matrix $A \in \R^{n\times n},$  $\tr(A)$ is the trace of the matrix, $\Vert A \Vert$ is the induced 2-norm, 
$\rho(A)$ denotes its spectral radius (\emph{i.e.} the  maximum of the absolute values of its eigenvalues), and $\sigma_{\max}(A)$ represents the maximum singular value of  $A.$ 
The uniform distribution with support $[a,b]$ is denoted by $\Uc(a,b), $ whereas the normal distribution with mean $\mu$ and covariance $\sigma^2$ is denoted by $\Nc(\mu, \sigma^2).$

\section{Background on Dynamic Programming}
\subsection{General formulation}
We start with the necessary background on the DP formulation of 
infinite-horizon, discrete-time, deterministic 
optimal control problems.
We refer the reader to \cite{bertsekas_volI,bertsekas_volII} for a comprehensive review on the DP theory.

We consider a deterministic discrete-time  dynamical system
\beq \label{eq:sys}
x_{t+1}  = f(x_t, u_t),
\eeq
where $x_t \in \Xb \subseteq \R^n$ represents the state of the system at time $t\in \mathbb{N}$, $u_t \in \Ub \subseteq \R^m$ the control action, 
and $ f : \Xb \times \Ub \to \R $ describes the evolution of the system.
Whenever convenient, we will use the shorthand $x_u^+$ to denote  $f(x,u).$ 
We restrict attention to \emph{stationary deterministic feedback policies} $\Pi = \{\pi \; | \; \pi :  \Xb \to \Ub\}.$ A non-negative \emph{stage-cost} function $l:  \Xb \times \Ub \to \R_{+}$ is associated to each state-action pair. 
Given the initial state $x$, the infinite-horizon $\gamma$-discounted cost for a policy $\pi$ is  
\beq \label{eq:cost_pi}
v_\pi (x) =  \sum_{t=0}^{\infty} \gamma^t l(x_t, \pi(x_t) ),
\eeq
where  $ x_0 = x  $, $x_{t+1} = f(x_t, \pi(x_t))$ for any $t$, and $\gamma \in (0,1]$ is a discount factor. We say that the policy $\pi$ is \emph{proper} if $v_\pi(x) < \infty$ for all $x \in \Xb$. In this case $v_\pi$ is the unique solution to the Bellman equation \cite[Corollary 1.2.2.1]{bertsekas_volII}
\beq \label{eq:Bell_pi}
   v_\pi (x) = l(x, \pi(x)) + \gamma v_\pi (f(x,  \pi(x) ),    \qquad \forall x \in \Xb.
\eeq
The aim of the optimal control problem is to find an optimal policy $\pi^* $ which minimizes the cost \eqref{eq:cost_pi},
\beq \label{eq:cost_opt}
 v^*(x) := v_{\pi^*} (x)  = \inf_{\pi \in \Pi} v_\pi (x)  \qquad \forall x \in \Xb,
\eeq
where $ v^* (x) $ is the optimal cost function. Throughout the paper, we work under the following standard assumptions:
\begin{assumption}\label{ass:stabilizable}
The  stage cost function $l(x,u)$ is non-negative and inf-compact, the system dynamics $f(x,u)$ is continuous, and a proper policy $\pi(x)$ exists. 
\end{assumption}
These assumptions ensure that the problem is well-posed, in the sense that $v^* \in \Sc(\Xb)  $, $\pi^*$ is measurable and the infimum in \eqref{eq:cost_opt} is attained \cite[Lemma 2.20]{hernandez2023introduction}.

The optimal value function $v^*$ is the unique solution to the \emph{Bellman optimality equation} \cite[Proposition 1.2.2]{bertsekas_volII}
\beq \label{eq:Bell_opt}
   v^* (x) = \inf_{u \in \Ub} \{ l(x, u ) + \gamma  v^* (f(x , u)) \}  ,    \qquad \forall x \in \Xb.
\eeq
If $v^*$ is known, then an optimal policy can be derived as
$$
\pi^*(x) = \argmin_{u \in \Ub } \left \{ l(x,u) + \gamma v^*(f(x,u)) \right \} ,    \qquad \forall x \in \Xb.$$
However, evaluating an optimal policy requires not only the availability of the optimal value function $v^*$, but also the dynamics function $f$ and the stage cost function $l$.
If $f$ or $l$ are not known, then $v^*$ is not sufficient for deriving an optimal policy. In such cases, one can introduce the \textit{q}-function \cite{Watkins1992} associated to the policy $\pi$,
\beq \label{eq:q_pi_def}
q_\pi (x,u) := l(x,u) + \gamma 	v_\pi ( f(x,u) ),  \quad \forall (x,u) \in \Xb \times \Ub.
\eeq
This can be interpreted as the cost of playing the input $u$ at the initial state $x$ and the policy $\pi$ thereafter. The \textit{q}-function also satisfies a Bellman equation,
\beq \label{eq:Bell_q_pi}
    q_\pi (x,u) = \underbrace{l(x,u) +  \gamma  q_\pi( f(x,u) , \pi( f(x,u) )  ) }_ {:= (F_\pi q_\pi)(x,u) },
\eeq 	
$\forall (x,u)\in \Xb \times \Ub$, and the optimal \textit{q}-function $q^*$ is defined as
$$ 
 q^*(x,u) := q_{\pi^*} (x,u) = l(x,u) + \gamma  v^* (f(x,u)) \quad 
$$
$\forall (x,u) \in \Xb \times \Ub.$ Thus $q^*$ is the unique solution of the \emph{Bellman optimality equation for q-functions},
\beq \label{eq:Bell_q_opt}
q^* (x,u)  :=  \underbrace {l(x,u) + \gamma 	\inf_{w \in \Ub}  q^*( f(x,u), w)  }_{:= (Fq)(x,u) } ,  
\eeq 
$\forall (x,u) \in \Xb \times \Ub.$ The operators $F_\pi, F : \Sc(\Xb) \to \Sc(\Xb)$ are monotone, $\gamma$-contractive with respect to the sup-norm, and have unique fixed points $q_\pi$, $q^*$, respectively \cite{cogill2006}. The advantage of the \textit{q}-function reformulation is that, once $q^*$ is known, the optimal policy $\pi^*$ can be derived by  
$$ \pi^*(x) = \argmin_{u \in \Ub } q^* (x,u) , \qquad \forall x \in \Xb,$$
without the knowledge of $f$ and $l.$

\subsection{Linear quadratic problem} \label{sec:back_LQR}
We consider the $\gamma$-discounted, deterministic, linear quadratic regulator (LQR) problem where the system is governed by   
the linear dynamics 
\beq \label{eq:LTI_sys}
x_{t+1} = f(x_t, u_t) = A x_t+ B u_t ,
\eeq
and it is subject to a quadratic stage cost 
\beq \label{eq:cost_quad}
l(x,u) =  \begin{bmatrix} x \\ u \end{bmatrix} \tp L \begin{bmatrix} x \\ u \end{bmatrix}, \quad L \in \Sbb^{m+n}.
\eeq
The matrix  $L = \m{L_1 & L_{12} \\ L_{12}\tp & L_2},$ with $L_1 \in \Sbb^n$ and $L_2 \in \Sbb^m ,$ is  positive semidefinite and  $L_2 \succ 0.$ 
We consider linear policies of the form 
\beq
\label{eq:lin_pol}
 \pi(x) =  K x, \quad K \in \R^{m \times n}.
\eeq
The next result follows from standard linear system theory arguments \cite{antsaklis1997linear}.
\begin{propo}\label{propo:stability}
    A linear policy $\pi(x) =  K x $  is proper if and only if  the spectral radius of the closed-loop matrix $\rho( A+BK)$ is smaller than $1/ \sqrt{\gamma}.$ Hence, Assumption \ref{ass:stabilizable} is satisfied if and only if the pair $ (\sqrt{\gamma} A, \sqrt{\gamma} B )$ is stabilizable. 
\end{propo}

The fixed point $q_\pi$ of the Bellman operator $F_\pi$ is \cite{bradtke1992reinforcement}
\beq \label{eq:qpi_LP} q_\pi (x,u) = \m{x \\ u}\tp Q_\pi \m{x \\u},  \quad  Q_\pi \in \mathbb{S}^{m+n}, \eeq
where the matrix $Q_\pi$ solves the Lyapunov equality
\beq \label{eq:lyap_Qpi}
Q_\pi =  L + \gamma  A_K \tp  Q_\pi A_K ,
\eeq
with
\beq \label{eq:AK_def} A_K = \m{A & B \\ KA & KB}. \eeq 
Since the spectral radius of $A_K$ is identical to that of $A + BK$ \cite[Lemma1]{Lu_TAC2019}, then from Proposition \ref{propo:stability} we conclude that $\sqrt{\gamma} A_K $ is asymptotically stable
and the Lyapunov equation \eqref{eq:lyap_Qpi} admits a unique positive definite solution. Finally, we point out that also the fixed point of $F$ is a quadratic form given by \eqref{eq:qpi_LP} with $K=K^*$, where $K^*$ is the optimal control gain \cite{bradtke1992reinforcement}. Moreover, even in the case of affine dynamics, additive noise, or generalized quadratic stage costs, the fixed points of $F_\pi$ and $F$ belong to the space of generalized quadratic functions \cite{MartinelliAffine}. Here, for the sake of simplicity, we focus on LQR problems described by \eqref{eq:LTI_sys}-\eqref{eq:lin_pol}.

\section{The Linear Programming approach to Dynamic Programming}
We introduce the LP approach to solve the Bellman Equations \eqref{eq:Bell_q_pi}-\eqref{eq:Bell_q_opt}
and we analyse the feasible regions of the resulting LPs.

\subsection{Exact linear program for $q_\pi$} \label{sec:exactLP_pi} 

The monotonicity and contractivity properties of the operator $F_\pi$ imply that any $q \in \Sc(\Xb \times \Ub) $ satisfying  
$ q \leq F_\pi q  $
is a lower bound for $q_\pi$.
Maximising among these lower bounds leads to the  infinite-dimensional LP 
    \beq \label{pb:LP_q_pi} 
    \begin{aligned} 
			& \sup_{q \in \Sc(\Xb \times \Ub)  }  \; \int_{\Xb \times \Ub}  q(x,u) c(dx, du) \\
			& \text{s.t.}  \quad   q(x,u) \leq  F_\pi q (x,u) \;\;  \forall  (x,u) ,
    \end{aligned} 
    \eeq
where $c \in \Mc^+$ is a finite non-negative measure over $\Xb \times \Ub$.	  
\begin{propo}	\cite{banj2019}
    The function $q_\pi$ coincides for $c$-almost all $(x,u) \in  \Xb \times \Ub $ with a maximizer of the LP problem \eqref{pb:LP_q_pi}.
\end{propo}

In the LP literature, $c(\cdot)$ is typically selected to be a probability measure \cite{Beuchat2016PerformanceGF,de2003linear}.
For example, if the state-action space is unbounded one can use a Gaussian distribution, or if it is compact a uniform distribution.

Next, we establish some properties of the feasible region
$$
\mathcal{F}_{\pi} := \left\{  q \in \Sc(\Xb \times \Ub)   :  q(x,u) \leq F_\pi q (x,u)  \;\;    \forall  (x,u) \right\}.
$$

\begin{propo}\label{propo:feasible_LP_pi}
The set  $\mathcal{C}_{\pi}:= \mathcal{F}_{\pi} - q_{\pi}$ is a salient convex cone (see Appendix \ref{app:convexcone}  for the definitions).
Moreover, for any $\bar q \in \mathcal{C}_{\pi}  $ and for any $ \lambda \geq 0$,  $\bar q - \lambda \in \mathcal{C}_{\pi}.$
\end{propo}

\begin{proof}
Convexity of $\mathcal{C}_{\pi}$  easily follows   from the fact that the inequality constraint defining $\Fc_\pi$ is affine, thus convex,  in $q.$   For the rest of the proof, consider an arbitrary $q \in \mathcal{F}_{\pi}$ and let $\bar q =  q -  q_{\pi} \in \mathcal{C}_{\pi}.$ 
To show that $\mathcal{C}_{\pi}$ is a cone, 
we have to show that $\lambda \bar q \in \mathcal{C}_{\pi}$ for any $\lambda \geq 0.$ 
Note that   
$ \lambda \bar q = (  \lambda q + (1-\lambda)  q_{\pi} )  - q_{\pi} $ 
belongs to $\mathcal{C}_{\pi}$ if and only if $\lambda q + (1-\lambda)  q_{\pi} \in \mathcal{F}_{\pi}.$
Since $q \in \mathcal{F}_{\pi}$ and $q_\pi$ satisfies  \eqref{eq:Bell_q_pi}, we have that 
\begin{align*}
&    \lambda q(x,u) +  (1-\lambda)  q_{\pi} (x,u)  
\leq  \lambda \left( l(x,u) +  \gamma  q(x_u^+, \pi(x_u^+)  )   \right)  
   \\ 
&  + (1-\lambda) \left( l(x,u)  + \gamma  q_{\pi}(x_u^+, \pi(x_u^+)  )   \right) 
   \\
& = l(x,u) + \gamma \left(  \lambda q(x_u^+, \pi(x_u^+)) + (1-\lambda)  q_{\pi}(x_u^+, \pi(x_u^+)) \right)   
\end{align*}
for any $(x,u).$ We conclude that  $\lambda q + (1-\lambda)  q_{\pi} \in \mathcal{F}_{\pi}$, hence $ \lambda \bar q \in \mathcal{C}_{\pi}.$ 
Recall that the cone $\mathcal{C}_{\pi}$ is salient if and only if  $ - \bar q \in \mathcal{C}_{\pi} $ implies that  $\bar q $ is the identically zero function.
From $q \in \mathcal{F}_{\pi}$ and the Bellman equality \eqref{eq:Bell_q_pi}, we obtain that 
\beq \begin{aligned} \label{eq:proof_salient1}
&-q(x,u) + 2q_\pi(x,u) \geq -l(x,u) - \gamma  q(x_u^+, \pi(x_u^+) )   \\
&+ 2 l(x,u) + 2 \gamma  q_\pi (x_u^+, \pi(x_u^+) )    \\ 
&= l(x,u) + \gamma \left(- q(x_u^+, \pi(x_u^+)) +  2q_\pi (x_u^+, \pi(x_u^+)) \right)
\end{aligned} \eeq
for any $(x,u).$
If $ -\bar q =   -q +  q_{\pi}  = ( -q +2q_{\pi}) -q_\pi$ belongs to  $\mathcal{C}_{\pi},$ then $(-q + 2q_\pi)\in \mathcal{F}_{\pi}$, that is 
\beq \begin{aligned}\label{eq:proof_salient2}
    -q(x,u) + 2q_\pi(x,u)   \leq & 
    l(x,u) + \gamma  \big(- q(x_u^+, \pi(x_u^+)) \\
    &  +  2q_\pi (x_u^+, \pi(x_u^+))\big).  
\end{aligned}
\eeq
Inequalities \eqref{eq:proof_salient1} and \eqref{eq:proof_salient2} hold if only if  
\begin{align*}
-q(x,u) + 2q_\pi(x,u) 
= &
l(x,u) + \gamma \big(- q(x_u^+, \pi(x_u^+))  \\
 &+  2q_\pi (x_u^+, \pi(x_u^+)) \big)  &{}
\end{align*}
for any state-action pair. This means that $-q+ 2q_\pi$ is a fixed point of the Bellman policy operator $F_\pi.$
However, since $q_\pi$ is the unique fixed point of the operator $F_\pi,$ we must have $-q + 2q_\pi = q_\pi.$ Hence $q = q_\pi$ and $\bar q = q - q_\pi$ is identically zero.
Finally we show that  $\bar q - \lambda \in \mathcal{C}_{\pi}$ for any $\lambda \geq 0.$
Indeed,
\begin{align*}
q(x,u) - \lambda   &\leq l(x,u) + \gamma  q(x_u^+, \pi(x_u^+)  )  - \lambda  \\
&\leq l(x,u) + \gamma  \left( q(x_u^+, \pi(x_u^+)  )   - \lambda \right)  
\end{align*}
where the first equality holds because $q \in \mathcal{F}_{\pi}$  and the second because  $\lambda \geq 0$ and 
$\gamma \in (0,1 ] .$
\end{proof}
Proposition \ref{propo:feasible_LP_pi} states that $\mathcal{F}_{\pi}$ is an affine convex cone with apex $q_\pi$.
This  gives an intuitive explanation of the reason why $ q_\pi$ is a maximizer of the LP \eqref{pb:LP_q_pi} for all non-negative measures: since non-negative measures always push towards one direction in the non-negative orthant, they force the apex of the cone to be an optimizer.

\subsection{Approximate linear program for $q_\pi$} \label{sec:appLP_pi}

Computing the q-function by solving the optimization problem \eqref{pb:LP_q_pi} is in general computationally intractable for the following reasons: 
\begin{enumerate}	
	\item The optimization variable $q(x,u)$ lies in an infinite dimensional space $\Sc (\Xb \times \Ub);$
	\item The number of constraints is infinite since the inequalities need to be satisfied for all $(x,u) \in \Xb \times \Ub.$
\end{enumerate}
To overcome these sources of intractability, the original infinite problem is approximated by a finite one.
This approach is referred to as ADP \cite{cogill2006, Beuchat2016PerformanceGF, deFarias2004}. 

The first approximation involves restricting the q-functions to a finite dimensional linear subspace of $\Sc$ as suggested in \cite{SCHWEITZER1985568}.
Specifically, we assume that the q-functions belong to the span of a finite family of linearly independent functions $\Phi_i(x,u)$ for $i = 1,2\dots, r,$ 
and we parameterize the restricted function space as  
$$ \Sc^{\Phi} (\Xb \times \Ub) : =  \left \{ \sum_{i=1}^{r} \alpha_i \Phi_i(x,u) =  \Phi(x,u) \alpha\; | \; \alpha_i \in \R \right\}. $$
Here $\alpha = \m{\alpha_1 & \dots & \alpha_r}\tp$ is the parameters vector and 
$\Phi(x,u) = \m{\Phi_1(x,u) & \dots & \Phi_r(x,u)}$ is the   features  (row) vector 
evaluated at the point $(x,u).$
Substituting $ \Sc $ with $ \Sc^{\Phi}$ in \eqref{pb:LP_q_pi} leads to a semi-infinite approximation of the original problem.
Important classes of approximation for q-functions are based on radial basis functions, quadratic functions
and sum-of-squares polynomials \cite{Sutton1998, bertsekas_volII, summers2013sos}.
We note that, while for any positive measure $c$ solving the problem \eqref{pb:LP_q_pi} yields $q_\pi$, this is no longer the case for the approximate LP. 
When restricting the functions to a subspace, the choice of state-action relevance weights $c$ may have a significant impact
on the quality of the resulting approximation. Indeed, the measure $c$ can be interpreted as allocating approximation quality over the state-action space \cite{de2003linear, Beuchat2016PerformanceGF}. 

For the infinite number of constraints, if the dynamics $f$ and the cost $l$ are available and have a certain structure, and the   features vector $ \Phi$  is selected appropriately, 
the constraint set can sometimes be represented exactly or relaxed via the $\Sc$-procedure or sum-of-squares \cite{wang2015approximate, summers2013sos}.
However, if we assume that the dynamics is unknown, we cannot approximate the feasible set using these tools. 
In this case it is possible to relax the constraints via sampling \cite{deFarias2004,John2017}.
If we measure the state $x_i \in \Xb$, apply the input $u_i \in \Ub$, and observe the new state $x_i^+ = f(x_i, u_i)$ and the incurred cost $l_i = l(x_i, u_i),$ we can obtain a set of tuples  $D =\{(x_i,u_i,x_i^+,l_i)\}_{i=1}^{N}$. Note that we use index $i$ instead of $t$ to emphasize that data tuples do not have to be consecutive in time.  
If we consider only a finite number of constraints, one for each tuple, we obtain the following data-driven finite dimensional approximation of the LP \eqref{pb:LP_q_pi}
\beq \label{pb:LP_q_pi_data} 
    \begin{aligned} 
     &\tilde{J}_{\pi} = \sup_{q \in \Sc^{\Phi} }  \; \int_{\Xb \times \Ub} q(x,u) c(dx, du) \\
     &\text{s.t.}  \quad   q(x_i,u_i) \leq l_i + \gamma  q(x_i^+, \pi(x_i^+))  \quad   i=1,\dots N.
    \end{aligned} 
\eeq
The constraint set 
\begin{align*}
\tilde{\mathcal{F}}_{\pi} := \{ q \in \Sc^{\Phi} \; | \; q (x_i,u_i) \leq F_\pi q(x_i, u_i) , \; i=1, \dots , N \}
\end{align*}
will of course implicitly depend on the unknown dynamics $f$ and the cost function $l$, but once the dataset in known, 
the explicit dependence is not needed.

\begin{assumption}\label{ass:basis_fcn}
    The set $ \Sc^{\Phi}$ contains the fixed point $q_\pi$  of the Bellman equation \eqref{eq:Bell_q_pi}.
\end{assumption}
While for LQR problems this assumption is easy to satisfy by restricting attention to the space of quadratic functions, for most other problems it is difficult to choose the right $\Phi$ such that this hypothesis holds.
Nonetheless, since below we focus on the effect of constraint sampling, rather than of q-function parameterization, it is convenient to work under this simplifying assumption. 
 \begin{assumption}\label{ass:full_rank}
    The data matrix 
    \beq\label{eq:def_datamatrix}
    M= \m{ 
    	\Phi(x_1,u_1)-\gamma \Phi(x_1^+,\pi(x_1^+)) 
    	\\ 
    	\vdots
    	\\
    	\Phi(x_N,u_N)-\gamma \Phi(x_N^+,\pi(x_N^+)) 
    } \in \R^{N \times r} 
    \eeq
    has full column rank.
 \end{assumption}
Assumption \ref{ass:full_rank} can be thought of as a \emph{persistence of excitation} requirement \cite{goodwin2014adaptive,de2019formulas}. It implies that $ N \geq r$ and $r$ constraints in Problem \eqref{pb:LP_q_pi_data} are linearly independent. 
Note that the rank of $M$ depends on the choice of the state-input pairs $(x_i,u_i),$ but also on the choice of $\Phi.$

\begin{propo}\label{propo:LP_q_pi_data}	
The set $\tilde{\mathcal{C}}_\pi:= \tilde{\mathcal{F}}_\pi  - q_{\pi}$  is a polyhedral salient convex cone.
Moreover, for any $\bar q \in \tilde{\mathcal{C}}_\pi$ and $\lambda \geq 0$, if $\bar q  - \lambda \in \Sc^{\Phi}$ then $\bar q  - \lambda \in \tilde{\mathcal{C}}_\pi$
\end{propo}
\begin{proof}
The proof that  $\tilde{\mathcal{C}}_\pi$ is a convex cone follows arguments similar to the ones of Proposition \ref{propo:feasible_LP_pi}. 
To prove that the cone is salient we show that if
$ \bar q \in  \tilde{\mathcal{C}}_\pi$ and $- \bar q \in  \tilde{\mathcal{C}}_\pi,$ then $\bar q $ is the identically zero function.
Indeed, using the same reasoning as in Proposition \ref{propo:feasible_LP_pi}, we obtain that
\begin{align*}
& - \Phi(x_i,u_i)\alpha  + 2\Phi(x_i,u_i) \alpha_\pi   =  
	  \\ 
&	 l_i -\gamma   \Phi(x_i^+, \pi(x_i^+)) \alpha
	+2 \Phi(x_i^+, \pi(x_i^+)) \alpha_\pi 
\end{align*} 
for $i = 1, \dots, N$. 
Equivalently, 
$
M (2\alpha_\pi - \alpha ) = \m{l_1 & \dots & l_N}\tp.
$ 
This is a linear system of equations which has the unique solution $2\alpha_\pi - \alpha = \alpha_\pi$ under the assumption that $M$ is full-rank.
We conclude that $\alpha = \alpha_\pi $ and therefore  $\bar q = 0.$  
The last statement of the theorem is proved following the same steps of the analogous claim in Proposition \ref{propo:feasible_LP_pi}. 
\end{proof}

The effect of the constraint sampling is to construct a polyhedral cone approximation 
of the infinitely-generated cone 
$ \{ q \in \Sc^{\Phi} \; | \; q (x,u) \leq F_\pi q(x, u) \, , \forall (x,u) \}.$
Specifically, each scenario $(x_i, u_i, l_i, x_i^+)$ restricts the feasible region to a half-space of the $S^{\Phi}$-space, separated by the hyperplane 
$  \left(\Phi(x_i,u_i) - \gamma  \Phi(x_i^+,\pi (x_i^+) \right) \alpha -  l(x_i, u_i) = 0,$
which passes through $q_\pi$ and is tangent to the infinitely-generated cone.
Because of this polyhedral approximation, $\tilde{\mathcal{F}}_\pi$  might be unbounded in the direction of growth of some  non-negative measures,  as shown in the numerical simulations of Section \ref{sec:numerics}. 
The assumption that $\Sc^{\Phi}$ contains $q_\pi$ is crucial;
otherwise there is no guarantee that there is a function in $\Sc^{\Phi} $ that satisfies all the constraints with equality.

From the geometry of feasible region $\tilde{\mathcal{F}}_{\pi}$, it seems plausible that, under boundedness assumptions, 
the optimization problem \eqref{pb:LP_q_pi_data} can recover ``more often than not'' the solution $q_\pi$ exactly. We can envision three possible outcomes:
\begin{enumerate}
    \item $\tilde{J}_{\pi}  = \infty $
    \item $\tilde{J}_{\pi}  < \infty $  and there is a unique optimizer equal to ${q}_{\pi}.$
    \item $\tilde{J}_{\pi}  < \infty $ and there is an infinite number of optimizers which lie on the boundary of the feasible set $\tilde{\mathcal{F}}_{\pi}.$
\end{enumerate}
In the last two cases, the set of optimizers always includes ${q}_{\pi}.$ Indeed, since $ \tilde{\mathcal{F}}_{\pi} $ is a convex polyedral cone 
with unique extreme point ${q}_{\pi},$ by the Fundamental Lemma of Linear Programming \cite[Theorem 2.7]{bertsimas-LPbook}, ${q}_{\pi}$ is a solution of the LP problem.
If the solution is not unique, then it must lie on one face
of the feasible region, that happens to be orthogonal to the measure $c$ (see Appendix \ref{app:convexcone} for the definition of face).
This third case is degenerate in that, perturbing $c$ or the constraints would break the symmetry and lead to one of the previous cases. 
Hence, the statement that ``more often than not'' if the problem is bounded, the optimizer is equal to $ {q}_{\pi}.$  


Problem \eqref{pb:LP_q_pi_data} is always feasible, but there is nothing to ensure that $\tilde{J}_\pi$ is finite.
This is indeed a major challenge,  especially for large-scale systems.   
Next, we exploit \emph{duality theory} to derive a necessary and sufficient condition for the LP \eqref{pb:LP_q_pi_data}  to be bounded. 
\begin{theorem}\label{th:Farkas_pi}
The optimal solution to Problem \eqref{pb:LP_q_pi_data} is finite if and only if there exists a row vector $\lambda = \m{\lambda_1 & \dots & \lambda_N},$ with $\lambda_i \geq 0$ for $i=1, \dots, N,$ such that 
\beq \label{eq:dual}
    \int_{\Xb \times \Ub} \Phi(x,u) c(dx,du) =  \lambda M .
\eeq
\end{theorem}
\begin{proof}
The dual LP associated to \eqref{pb:LP_q_pi_data} is \cite[Ch.6]{matouvsek2007understanding}. 
\beq \label{eq:dual_Lp_pi_data}
\begin{aligned} 
 \sup_{\lambda }  \; \;  & \sum_{i=1}^{N} \lambda_i l_i  \\
\text{s.t.}  \; \;  &\lambda_i \geq 0,  \quad  i=1, \dots, N  \\
  & \int_{\Xb \times \Ub} \Phi(x,u) c(dx,du)   =  \lambda M.  
\end{aligned}  \eeq
From the \emph{duality theorem}, the primal problem \eqref{pb:LP_q_pi_data} is bounded if and only if the dual problem \eqref{eq:dual_Lp_pi_data} is feasible, concluding the proof.
\end{proof}

Leveraging on Theorem \ref{th:Farkas_pi}, we derive a sufficient condition to rule out  the degenerate outcome 3) above.
\begin{corollary} \label{coroll:bound_opt}
Let  
$M_i  = \Phi(x_i, u_i) - \gamma \Phi(x_i^+, u_i^+)$ for $i=1,\dots, N,$
and assume that   
\beq \label{eq:C_span}
 \int_{\Xb \times \Ub} \Phi(x,u) c(dx,du)  \notin \text{span} \{ M_{i_1}, \dots , M_{i_{r-1}} \}
\eeq
for any set $ \{i_1, \dots, i_{r-1} \} \subset\{1,2, \dots, N\}.$  
Then, if Problem \eqref{pb:LP_q_pi_data} is bounded, 
its optimizer is unique and coincides with the fixed point $q_\pi $ of the Bellman operator $T_\pi.$
\end{corollary}
\begin{proof}
    Let  $\tilde{q}_\pi(x,u) = \Phi(x,u) \tilde{\alpha}_\pi$ be an optimizer of \eqref{pb:LP_q_pi_data}. By Theorem \ref{th:Farkas_pi}, there exists a positive  vector $\lambda$ such that \eqref{eq:dual} holds.  Equations \eqref{eq:dual} and \eqref{eq:C_span} imply that there exists $\{i_1, \dots, i_{r} \} \subset\{1,2, \dots, N\} $ such that $\lambda_{i_1}, \dots, \lambda_{i_r}   $ are strictly positive with $\{ M_{i_1}, \dots M_{i_r} \} $ linearly independent.
    Then, from the complementary slackness conditions \cite[Section 5.5.2]{boyd} we have that 
    $ \m{M_{i_1}\tp & \dots  & M_{i_r}\tp  }\tp \tilde{\alpha}_\pi = \m{ l_{i_1} & \dots & l_{i_r}  } \tp. $
    Since $\m{M_{i_1}\tp & \dots & M_{i_r}\tp  }\tp $ is full rank and $q_\pi (x,u)= \Phi(x,u) \alpha_\pi $ is such that  $ \m{M_{i_1}\tp&  \dots & M_{i_r}\tp  }\tp \alpha_\pi = \m{ l_{i_1} & \dots & l_{i_r}  } \tp, $
    we conclude that $\tilde{\alpha}_\pi = \alpha_\pi. $  
\end{proof}

Theorem  \ref{th:Farkas_pi} and its Corollary \ref{coroll:bound_opt} will be exploited in Section \ref{sec:smartsampl}
to formulate a finite LP for $q_\pi$ with boundedness guarantees in the LQR setting.

\subsection{Linear programs for $q^*$} \label{sec:appLP_opt}
Next, we turn our attention to the optimal  q-function $q^*.$ 
\begin{propo}\cite{Beuchat2016PerformanceGF}
    The solution $q^*$ to the Bellman optimality Equation \eqref{eq:Bell_q_opt} coincides $c$-almost everywhere with a solution to the following convex optimization problem:
    \beq 
    \begin{aligned} \label{pb:LP_q_opt} 
			 \sup_{q \in \Sc(\Xb \times \Ub)  }  \; &\int_{\Xb \times \Ub}  q(x,u) c(dx, du) \\
			\text{s.t.}  \quad  & q(x,u) \leq Fq(x,u)   \; \;  \forall  (x,u) \in \Xb \times \Ub
    \end{aligned} 
    \eeq
    where $c\in \Mc^+ $ is a non-negative finite measure over $\Xb \times \Ub.$ 
\end{propo}
Note that, strictly speaking, Problem \eqref{pb:LP_q_opt} is not an LP because $F$ is a convex but nonlinear operator.
    However, it is possible to reformulate \eqref{pb:LP_q_opt} as an equivalent LP by dropping the infimum in $F$, and substituting the nonlinear constraint set with 
    $$	q(x,u) \leq l(x,u) +   q(x_u^+, w ) ,   \; \;  \forall  (x,u,w) \in \Xb \times \Ub^2. $$
Note that, in case of stochastic systems, the problem of obtaining a linear reformulation is not trivial and solutions have been proposed in \cite{cogill2006} and \cite{MARTINELLIAut}. Consider now the feasible set of the problem \eqref{pb:LP_q_opt}
$$
\mathcal F^* =\{ q \; \in \Sc(\Xb \times \Ub)  \; : \;    q(x,u) \leq Fq(x,u)   \; \;  \forall  (x,u) \}.
$$
\begin{propo}\label{propo:feas_q_opt}
    The set $ \mathcal F ^*$ satisfies:
    \begin{enumerate}
        \item $\Fc ^*$ is a convex set with $q^*$ as an extreme point.
        \item For any $q \in \Fc ^*$ and $\lambda \geq 0,$ then $q - \lambda \in  \Fc ^*$
        \item For all policies $\pi$, $\Fc^* \subseteq \Fc_\pi.$
    \end{enumerate}
\end{propo}

\begin{proof}
Part 1 is straightforward since the constraint in \eqref{pb:LP_q_opt} is convex in $q$.
Moreover, $q^*$ is clearly an extreme point of the set as it satisfies all the constraints with equality.  Part 2 follows the same steps of the corresponding claim in Proposition \ref{propo:feasible_LP_pi}. 
Finally, for Part 3, take $q \in \Fc^*$ and an arbitrary policy $\pi.$
Since $ q(x,u) \leq l(x,u) +   q(x_u^+, w ) $  for any $(x,u,w)$, then in particular the inequality holds for $w = \pi (x_u^+).$ Therefore
$  q(x,u) \leq l(x,u) +   q(x_u^+, \pi (x_u^+)) 
$
for any $(x,u)$ and $q \in \Fc_\pi.$ 
\end{proof}

The first property implies that the feasible set $\Fc^*$ is a convex set, while the second that it is unbounded.
However, as we will see in the following section, the set is not necessarily a cone. Property 3 establishes that, tough not a cone, $\Fc^*$ is contained in the intersection of the conic sets $\Fc_\pi$ for all policies $\pi.$ 

Problem \eqref{pb:LP_q_opt} is computationally intractable. 
Following the same reasoning of Subsection \ref{sec:appLP_pi}, we now provide a finite approximation of Problem \eqref{pb:LP_q_opt}.
We collect a dataset $D = \{x_i,u_i, x_i^+, l_i\}_{i=1}^{N}$  and 
we restrict the problem to a linearly parameterized set of functions  $\Sc^{\Phi}.$
\begin{assumption}\label{ass:basis_fcn_qopt}
The set of functions $\Sc^{\Phi} (\Xb \times \Ub)  $ contains the fixed point $q^*$  of the Bellman equation \eqref{eq:Bell_q_opt}.
\end{assumption}
Then, a data-driven reformulation of  \eqref{pb:LP_q_opt} is given by
\beq 
\begin{aligned} \label{pb:LP_q_opt_data} 
    \sup_{q \in\Sc^{\Phi} (\Xb \times \Ub) }  \;  &\int_{\Xb \times \Ub} q(x,u) c(dx, du) \\
    \text{s.t.}  \qquad    & q(x_i,u_i) \leq Fq(x_i, u_i)    \; \;  i = 1, \dots, N
\end{aligned} 
\eeq
Let $\tilde{\Fc}^*$ be its feasible region:
$$ \tilde{\Fc}^*:= \{ q \in\Sc^{\Phi} \;  | \; q(x_i, u_i) \leq Fq(x_i,u_i), \; i= 1, .. , N \}. $$
The set $\tilde{\Fc}^*$ inherits the properties of $\Fc^*$:
\begin{propo} \label{propo:set_ddLP_opt}
    The feasible set $\tilde{\Fc}^* $ satisfies:
    \begin{enumerate}
        \item $\tilde{\Fc}^*$ is a convex set with $q^*$ as an extreme point.
        \item For any $q \in \tilde{\Fc}^*$ and $\lambda \geq 0,$ if  $q - \lambda \in  \Sc^{\Phi} $ then $q - \lambda \in \tilde{\Fc}^*.$
        \item For all policies $\pi$, $\tilde{\Fc}^* \subseteq \tilde{\Fc}_\pi$ if problems \eqref{pb:LP_q_pi_data} and \eqref{pb:LP_q_opt_data} are built with the same dataset $D$ and the same set of functions $S^\Phi.$
    \end{enumerate}
\end{propo}
Analogously to $\Fc^*$, $\tilde{\Fc}^*$ is not necessary a cone rooted at $  q^*.$ 

Boundedness of the optimization problem \eqref{pb:LP_q_opt_data} is addressed in \cite{CDC22_meyn},
where it is shown that, under an ergodicity assumption, if the covariance matrix associated with the basis function 
\begin{align*}
    \Sigma &:=  \left( \lim_{N \to \infty} \frac{1}{N} \sum_{i=0}^{N} \Phi(x_i, u_i)\tp \Phi(x_i, u_i) \right)  \\
     & -\left( \lim_{N \to \infty} \frac{1}{N} \sum_{i=0}^{N} \Phi(x_i, u_i) \right)\tp \left( \lim_{N \to \infty} \frac{1}{N} \sum_{i=0}^{N} \Phi(x_i, u_i) \right)  
\end{align*} 
is full rank, then the constraint region $\tilde{\Fc}^* $ is bounded for $N$ sufficiently large; hence the optimal cost is finite \cite[Theorem 2.2]{CDC22_meyn}.
We note that the full-rank property on the covariance matrix $\Sigma$ involves an implicit assumptions on the choice of $\Phi(x,u).$  Indeed, in view of point 2 of Proposition \ref{propo:set_ddLP_opt}, if the finite-dimensional function space $\Sc^{\Phi}$ is such that $q - \lambda \in  \Sc^{\Phi} $ for any $q \in \Sc^{\Phi} $  and any $\lambda \geq 0,$ then the feasible region $\tilde{\Fc}^*$ is always unbounded and the corresponding matrix $\Sigma$ is rank-deficient.
In addition, the rank of $\Sigma$ also depends on the properties of the underlying system \eqref{eq:sys}. 
Indeed, in the following section, we discuss systems for which the feasible region is unbounded even with an infinite number of constraints. Also in this case, the covariance matrix $\Sigma$ will be rank-deficient.
Finally, we point out that the boundedness result holds only if the data tuples 
$\{x_i,u_i, x_i^+ \}_{i=1}^{N}$ come from one contiguous trajectory of the system and for a sufficiently large number of data points. 

\section{Linear Quadratic Regulator Problem}
\begin{figure} 
	\centering
	\includegraphics[width=\linewidth]{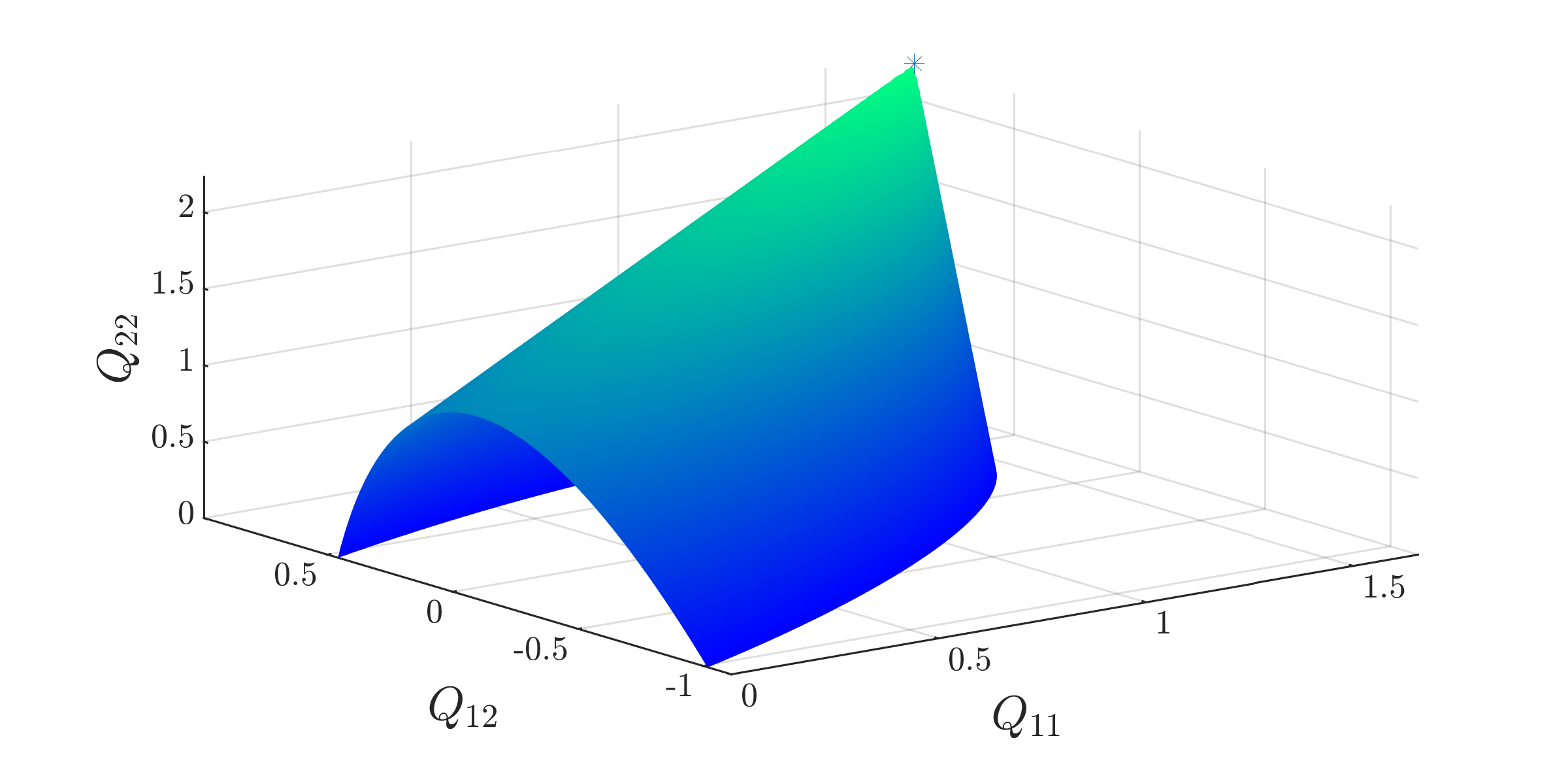}
	\caption{LQR problem: boundary of the feasible region $\Fc_\pi$ defined by \eqref{eq:LQR_pi_cone} for a system with  $n = 1$ and  $m = 1.$ }
	\label{fig:LQR_cone}
\end{figure}
We formulate and investigate the properties of the LPs for $q_\pi$ and $q^*$ in the LQR setting.

\subsection{Exact linear programs}

Consider the LQR problem presented in Subsection \ref{sec:back_LQR}.
Since the functions $q_\pi$ and $q^*$ are quadratic, we can restrict our attention to the space $\Sc^\Phi $ of  functions
\beq \label{eq:q_quad}
q(x,u) = \m{x\\u }\tp Q  \m{x\\u }, \quad Q\in \Sbb^{m+n},
\eeq 
which has dimension $r= 0.5(m+n)(m+n+1).$
Note that this choice of $\Sc^\Phi $ satifies Assumption \ref{ass:basis_fcn} and \ref{ass:basis_fcn_qopt}.

Now, consider the exact LP \eqref{pb:LP_q_pi} for $q_\pi.$
By substituting \eqref{eq:LTI_sys}- \eqref{eq:lin_pol} and \eqref{eq:q_quad} into \eqref{pb:LP_q_pi}, the feasible set becomes
\beq \label{eq:LQR_pi_cone} \Fc_\pi = \{Q \in \Sbb^{m+n} \; : \; L + \gamma A_K\tp Q A_K - Q \succeq 0\} .\eeq
This is a \emph{Linear Matrix Inequality} (LMI) in $Q$, which defines a convex conic constraint, as expected by Proposition \eqref{propo:feasible_LP_pi}.
A representation of $\Fc_\pi$ for a system with $1$ state variable and $1$ input is given in Figure \ref{fig:LQR_cone}.
Moreover, by defining the measure $c(\cdot)$ as a probability distribution  with zero mean and covariance $C \succ 0,$
we can provide a closed-form expression for the objective function
$ \int_{\Xb \times \Ub} q(x,u) c(dx,du)  = \tr (QC). $
Thus, the exact LP \eqref{pb:LP_q_pi} for policy evaluation becomes
\beq \label{pb:LP_LQR_pi}
\begin{aligned}
     \sup_{Q \in \Sbb^{m+n} } \; &\tr(QC) \\
      \text{s.t.} \; \; & L + \gamma A_K\tp Q A_K - Q \succeq 0.
\end{aligned}
\eeq

Analogously, 
by substituting \eqref{eq:LTI_sys}- \eqref{eq:lin_pol} and \eqref{eq:q_quad} into the LP \eqref{pb:LP_q_opt} for $q^*,$
we obtain an explicit characterization of its feasible set as
\begin{align*}
 \m{ x \\ u }\tp Q \m{x \\ u} 
\leq \m{ x \\ u }\tp L \m{x \\ u} + \gamma \m{ Ax + Bu \\ w } \tp Q \m{ Ax + Bu \\w } &{} \\
=  \m{ x \\ u }\tp L \m{x \\ u} + \gamma \m{ x\\ u \\ w }\tp  \m{ A & B & 0\\ 0 & 0 & I }  \tp Q  \m{ A & B & 0\\ 0 & 0 & I }\m{ x\\ u \\ w }  &{}
\end{align*}
for all $(x,u,w) \in \Xb \times \Ub^2.$
Appending $w$ to all the vectors and padding the matrices $L$ and $Q$ with zeros suggests that the feasible set is given by the solution to the following LMI
\beq\label{eq:LMI_Q}
\underbrace{\m{L & 0 \\ 0 & 0} + \gamma \m{ A & B & 0\\ 0 & 0 & I }  \tp Q  \m{ A & B & 0\\ 0 & 0 & I } - \m{Q & 0 \\ 0 & 0} }_{=\LMI(Q)} \succeq 0.
\eeq
Inequality \eqref{eq:LMI_Q} clearly defines a convex set, but in general it is not a cone.  In fact, the set $ \Fc^* := \{ Q \in \mathbb{S} ^{m+n} \; | \; \LMI(Q) \succeq 0\}$ can be either bounded or unbounded, depending on the underlying system.
\begin{theorem}
 \label{propo:LQR_opt_bounded}
  Consider the linear system \eqref{eq:LTI_sys} where $\rho (A) < 1/\sqrt{\gamma}.$
The feasible region $\Fc^*$ is bounded if and only if the pair $(A,B)$ is reachable. 
\end{theorem}
The proof is deferred to Appendix \ref{app:proof}.

\subsection{Approximate linear program for $q_\pi$}\label{sec:smartsampl}
Given a dataset $D,$ we can approximate the LP \eqref{pb:LP_LQR_pi} as
\beq \label{pb:LP_LQR_pi_data} 
\begin{aligned} 
&\tilde J_{\pi} = \sup_{Q \in \mathbb{S}^{m+n}  }  \; \tr (QC)  \\
& \text{s.t.}  \; \;   \m{x_i \\ u_i}\tp Q \m{x_i \\ u_i} \leq l_i + \gamma  \m{x_i^+\\ K x_i^+ } \tp Q \m{x_i^+\\ K x_i^+ }  \;   i=1,..,N. 
\end{aligned} 
\eeq

Guaranteeing the existence of a bounded optimal solution to \eqref{pb:LP_LQR_pi_data}  is not an easy task.  
This difficulty is even more pronounced  for large-scale systems. If we randomly generate the data points $(x_i, u_i, l_i, x_i^+)$ and the system is high-dimensional, typically $\tilde J_{\pi} < \infty$ only for extremely large $N.$ 
Our objective is to provide a systematic method to properly sample the dataset $D$ and/or to choose the covariance matrix $C$ such that Problem \eqref{pb:LP_LQR_pi_data} is guaranteed to be  bounded.
We first specialize  Theorem \ref{th:Farkas_pi} and Corollary \ref{coroll:bound_opt} to the LQR problem.
\begin{corollary}\label{coroll:LQR}
Define, with a slight abuse of notation, the matrices 
$$
M_i = \m{x_i \\ u_i} \m{x_i \\ u_i}\tp - \gamma \m{x_i^+\\ K x_i^+}\m{x_i^+\\ K x_i^+ }\tp, \; i=1,..,N.
$$
1) Problem \eqref{pb:LP_LQR_pi_data} is bounded if and only if 
there exist $\lambda_i \geq 0$ for $i=1, \dots, N$ such that 
\beq  \label{eq:feas}
C = \sum_{i=1}^{N} \lambda_i M_i 
\eeq 
2) Assume that 
\beq \label{eq:LQR_bound_opt_LQR}
C \notin \text{span} \{ M_{i_1}, \dots , M_{i_{r-1}} \}
\eeq
for any set $ \{i_1, \dots, i_{r-1} \} \subset\{1, \dots, N\}.$ Then, if Problem \eqref{pb:LP_LQR_pi_data} is bounded, the  optimizer is unique and coincides with the solution $Q_\pi$ to the Lyapunov equality \eqref{eq:lyap_Qpi}. 
\end{corollary} 

We next propose two different strategies to build Problem \eqref{pb:LP_LQR_pi_data} with boundedness guarantees.

\subsubsection*{Algorithm 1}\label{sec:method1}
Corollary \ref{coroll:LQR} reveals that Problem \eqref{pb:LP_LQR_pi_data} has
as unique optimizer $Q_\pi$ for almost any choice of the
covariance matrix $C$, provided that the problem is bounded.
One can the consider selecting the covariance matrix $C \succ 0$ as a function of the dataset $D$ such that \eqref{eq:feas} holds, instead of fixing it a priori.
This is possible if and only if 
\beq \label{eq:LMIforC}
\begin{aligned} 
\sum_{i=1}^{N} \lambda_i M_i  & \succ 0 \\
 \lambda_i  & \geq 0 \qquad i=1,.., N  
 \end{aligned}
\eeq
has solution, leading to the following problem: choose the pairs $ (x_i, u_i)$  such that \eqref{eq:LMIforC} is feasible.

To solve this problem, first we assume that 
\beq\label{eq:contraction} 
\sigma_{\max} (A_K) < 1/\sqrt{\gamma}.
\eeq 
Then, by Lemma  \ref{lemm:lyap}, 
$ I - \gamma A_K I A_K\tp  \succ 0.$
Comparing this LMI with \eqref{eq:LMIforC}, we conclude that any $D$ such that 
$$
I = \sum_{i=1}^{N} \lambda_i \m{x_i \\ u_i} \m{x_i \\ u_i}\tp
$$
for some $\lambda_i \geq 0$ ensures feasibility of  \eqref{eq:LMIforC}. 
For instance, we can select $\m{x_i\tp & u_i\tp}=e_i\tp $ for $i=1, \dots, n+m$ to meet this condition.
If \eqref{eq:contraction} is not satisfied, we can perform a change of coordinates such that the matrix $A_K$ in the new basis satisfies \eqref{eq:contraction} (see Lemma \ref{lemm:lyap}).
To derive the change of basis matrix, we have to solve the Lyapunov inequality:
$ P - \gamma A_K P A_K\tp \succ 0.$
Since the system dynamics is unknown, this Lyapunov inequality must be solved in a data-driven fashion. 
To this end, we sample an initial  dataset $D_1 = 
\{(x_i, u_i, l_i, x_i^+ )\}_{i=1}^{r}$ such that Assumption \ref{ass:full_rank} holds.
Then, a solution to the Lyapunov inequality is given by 
\beq \label{eq:algo1_P} P = \sum_{i=1}^{r} p_i \m{x_i \\ u_i} \m{x_i \\ u_i}\tp \eeq
where $p_i \in \R $ satisfy
\beq \label{eq:batch1} \sum_{i=1}^{r} p_i \left( \m{x_i \\ u_i} \m{x_i \\ u_i} - \gamma \m{x_i^+ \\ Kx_i^+} \m{x_i^+ \\ Kx_i^+} \right) =  W  \eeq
for an arbitrary $W \in \mathbb{S}^{n+m}$ with  $W \succ 0.$ 

Once $P$ is known, the desired change of coordinates matrix is given by its square root $P^{1/2}.$
Now, consider the vectors 
\beq  \label{eq:algo1_D2} \m{ x_{i+r} \\ u_{i+r} } = P^{1/2} e_i, \quad i=1, \dots, n+m \eeq
and collect the batch of data 
$D_2 = \{( x_i, u_i, l_i, x_i^+) \}_{i=r+1}^{N} $
with $N=r + n + m.$
The dataset $D$ obtained by appending $D_2$ to $D_1$ guarantees that \eqref{eq:LMIforC} is feasible (for example $\lambda_i = 0$ for $i=1,..,r$ and 
$\lambda_i = 1$ for $i=r,..,N$ is a solution). 
Hence, we can set $C$ according to \eqref{eq:feas} where $\lambda_i$'s  satisfy \eqref{eq:LMIforC}. 

The overall procedure is summarized in Algorithm \ref{algo1}. 

\begin{theorem}
    The LP \eqref{pb:LP_LQR_pi_data} where the covariance matrix $C$ and the dataset $D$ are obtained from Algorithm \ref{algo1} is bounded.
\end{theorem}
The proof follows from the line of reasoning leading up to the theorem statement. 
Note that, to guarantee uniqueness of the solution to \eqref{pb:LP_LQR_pi_data}, we should verify the sufficient condition \eqref{eq:LQR_bound_opt_LQR} of Corollary \ref{coroll:LQR}. 
In practice, simulations show that if the LMI \eqref{eq:LMIforC} is solved numerically, the degenerate case in which there is an infinite number of optimizers  never happens. 
Therefore, we avoid testing condition \eqref{eq:LQR_bound_opt_LQR}, which would be computationally expensive.

To determine the computational complexity of Algorithm \ref{algo1}, we note that solving the LMI \eqref{eq:LMIforC} is the most demanding operation. 
By using the SeDuMi solver, this requires $\Oc ( N^2(N+n+m)^{2.5} + (N+n+m)^{3.5})  $ operations \cite{labit2002sedumi}. Since $N = \Oc \big( (n+m)^2 \big),$ we conclude that the computational complexity  is $\Oc\big( (n+m)^{9} \big). $

\begin{algorithm}
\caption{}
\begin{algorithmic}[1]
\REQUIRE{policy $K$ }
\STATE Collect a dataset $D_1$ satisfying Assumption \eqref{ass:full_rank}
\STATE Solve the linear system \eqref{eq:batch1} for a certain $W\succ0$
\STATE Compute $P$ according to \eqref{eq:algo1_P}
\STATE Collect the dataset $D_2$ according \eqref{eq:algo1_D2}
\STATE Build the dataset $D$ by appending $D_2$ to $D_1$
\STATE Given $D$, solve the LMI  \eqref{eq:LMIforC} and let $\lambda_i, \; i=1,..,N$ be a solution.
\STATE Define $C$ according to \eqref{eq:feas}
\ENSURE{dataset $D,$ covariance matrix $C$}
\end{algorithmic}
\label{algo1} \end{algorithm}

\subsubsection*{Algorithm 2}
Algorithm 1 treats the covariance matrix $C$ in Problem \eqref{pb:LP_LQR_pi_data} as an additional decision variable. 
However, in many situations it may be more appropriate to fix the measure $c$ a priori, since it
may have a significant impact on the solution to  \eqref{pb:LP_LQR_pi_data} if Assumption \ref{ass:basis_fcn} is not satisfied. In the following, to match the specifications of a control
design we propose a strategy to sample a dataset $D$ such that Problem \eqref{pb:LP_LQR_pi_data} is bounded for a given covariance $C.$
First we compute  $P\succ 0$ solving the Lyapunov equality $ C = P - \gamma A_K P A_K\tp$ in a model-free manner.
As seen in Algorithm \ref{algo1}, this can be done by sampling a dataset $D_1$  with $r$ data points  satisfying Assumption \ref{ass:full_rank}. Then, $P$ is given by \eqref{eq:algo1_P} and \eqref{eq:batch1} and $W=C.$
Let  $\{v_i\}_{i=1}^{n+m}$ be  an orthonormal basis of eigenvectors of $P$ and $d_i > 0$ the corresponding eigenvalues;  recall that $P$ is a real symmetric matrix so all eigenvalues are real and an orthonormal basis of real eigenvectors always exists.
We sample a second batch of $n+m$ data points  $D_2 = \{x_i, u_i, l_i, x_i^+\}_{i=r + 1}^{N}$ so that 
\beq \label{eq:algo2_D2} \m{x_{i+r} \\  u_{i+r} }  = v_i \quad i=1, \dots, n+m,\eeq 
and we denote by $D$ the dataset obtained by appending $D_2$ to $D_1$. 

One can see that, for the data set $D$
there exist $\lambda_i \geq 0,$  $i=1,\dots ,N$, such that  \eqref{eq:feas} is satisfied, hence the LP  \eqref{pb:LP_LQR_pi_data} is bounded; for example, take $\lambda_{r + i} = d_i$ for $i = 1, \dots, n+m$ and $\lambda_i = 0$ otherwise. 
However this choice of $D$ introduces a certain ``structure'' in the problem, in the sense that
$ C \in \text{span} \{ M_{r+1}, \dots, M_{N} \}. $
Thus, the sufficient condition for uniqueness of the solution given in Corollary \ref{coroll:LQR} is not satisfied. On the contrary, simulations show that  
 we often end up in the degenerate case where the measure $c$ is perpendicular to a face of the feasible region, causing the problem to have an infinite number of maximizers.
We can overcome this issue by slightly perturbing the matrix $C.$ 
Specifically, we consider the matrix
\beq \label{eq:algo2_Ctilde}
\tilde C = C + \underbrace{\sum_{i=1}^{N} \varepsilon_i M_i}_{:=\Delta C} 
\eeq
with $\varepsilon_i>0$ positive small constants.  
Appropriately choosing the values of $\varepsilon_i,$ we guarantee that $\tilde C$ is positive definite and at the same time the distance $ \Vert \tilde C  - C \Vert$  is small.
This perturbation breaks the ``structure'' of the problem without affecting the feasibility of \eqref{eq:dual_Lp_pi_data}. 
Note that, by linearity of the trace operator, solving the LP problem \eqref{pb:LP_LQR_pi_data} with respect to $\tilde C$ instead of $C$ corresponds to considering the objective function 
$ \tr (QC) + \tr(Q \Delta C).$
This can be interpreted as a regularization procedure, with the coefficients $\varepsilon_i$ controlling the importance of the regularization term.

\begin{theorem}
A matrix $C \in \Sbb^{m+n},$ $C \succ 0$ is given. 
Let the covariance $\tilde C$ and dataset $D$ be obtained from Algorithm \ref{algo2}. Then the LP \eqref{pb:LP_LQR_pi_data},  solved with the perturbed $\Tilde{C},$ is bounded.
\end{theorem}

The proof follows from the argument leading up to the theorem statement. The procedure is summarized in Algorithm 2. The algorithm 
2 does not require solving a LMI, which alleviates the computational requirements with respect to Algorithm \ref{algo1}.
The most expensive operation is solving the linear system of equations \eqref{eq:batch1} which requires $\Oc\big( (n+m)^{6} \big)$
operations.

\begin{algorithm}
\caption{}
\begin{algorithmic}[1]
\REQUIRE{policy $K,$ covariance matrix $C$}
\STATE Collect a dataset $D_1$ satisfying Assumption \ref{ass:full_rank}
\STATE Solve the linear system \eqref{eq:batch1} where $W=C$
\STATE Compute $P$ according to \eqref{eq:algo1_P}
\STATE Compute the eigenvectors of the matrix $P$ and collect the dataset $D_2$ according to \eqref{eq:algo2_D2}
\STATE Build the dataset $D$ by appending $D_2$ to $D_1$
\STATE Define $\tilde C$ according to \eqref{eq:algo2_Ctilde}
\ENSURE{dataset $D,$ perturbed covariance matrix $\tilde C$}
\end{algorithmic}
\label{algo2} \end{algorithm}


\section{Numerical examples} \label{sec:numerics}
We present two numerical examples to highlight the key aspects of the theory.
Example 1 gives evidence for the difficulty of obtaining a bounded LP \eqref{pb:LP_LQR_pi_data} by randomly generating the dataset, especially for large-scale systems.
Example 2 demonstrates that Algorithm \ref{algo1} and Algorithm \ref{algo2} do solve this unboundedness issue. Moreover, it compares the performance of the two algorithms with the regularized solution presented in \cite{regLP_CDC}.
\\
In all the simulations, we consider 
the LQR problem with discount factor $\gamma = 0.95$ and stage-cost matrix $L = \text{diag} ( I_n,  10^{-3} I_{m} ). $ 

\emph{Example 1:}
We perform a Monte Carlo simulation of $200$ experiments, where for each experiment:
\begin{enumerate} 
\item 
We randomly generate a LTI system \eqref{eq:LTI_sys} with $n$ states and $m$ inputs.
The diagonal elements of the matrix $A$ are equal to $0.5,$ whereas the off-diagonal entries of $A$ and the entries of the matrix $B$
are generated from a uniform distribution $\Uc (-0.3, 0.3).$
We define a stabilizing policy \eqref{eq:lin_pol} such that 
the closed-loop matrix $A+BK$ has $n$ real eigenvalues evenly spaced in the interval $[0.2,0.8].$
\item
We collect the dataset $D= \{x_i, u_i, l_i, x_i^+\}_{i=1}^{N}$ where $x_i \sim \Nc(0,1)$ and  $u_i \sim \mathcal{U} (-2, 2).$ 
\item We solve the LP \eqref{pb:LP_LQR_pi_data} with covariance $C = I.$
\end{enumerate} 
We are interested in evaluating the relative frequency $f$ of Monte Carlo experiments leading to bounded LP problems within the total number of experiments for different values of the system dimensions $(n,m)$ and different amount of data $N.$
The results are summarized in Figure \ref{fig:bounded}. 
The figure shows that, if we randomly generate the data, the number of constraints required to bound the cost of  \eqref{pb:LP_LQR_pi_data} grows rapidly with the system dimension. This issue is particularly  severe for large-scale systems.  

\begin{figure} 
\centering
\includegraphics[width=\linewidth]{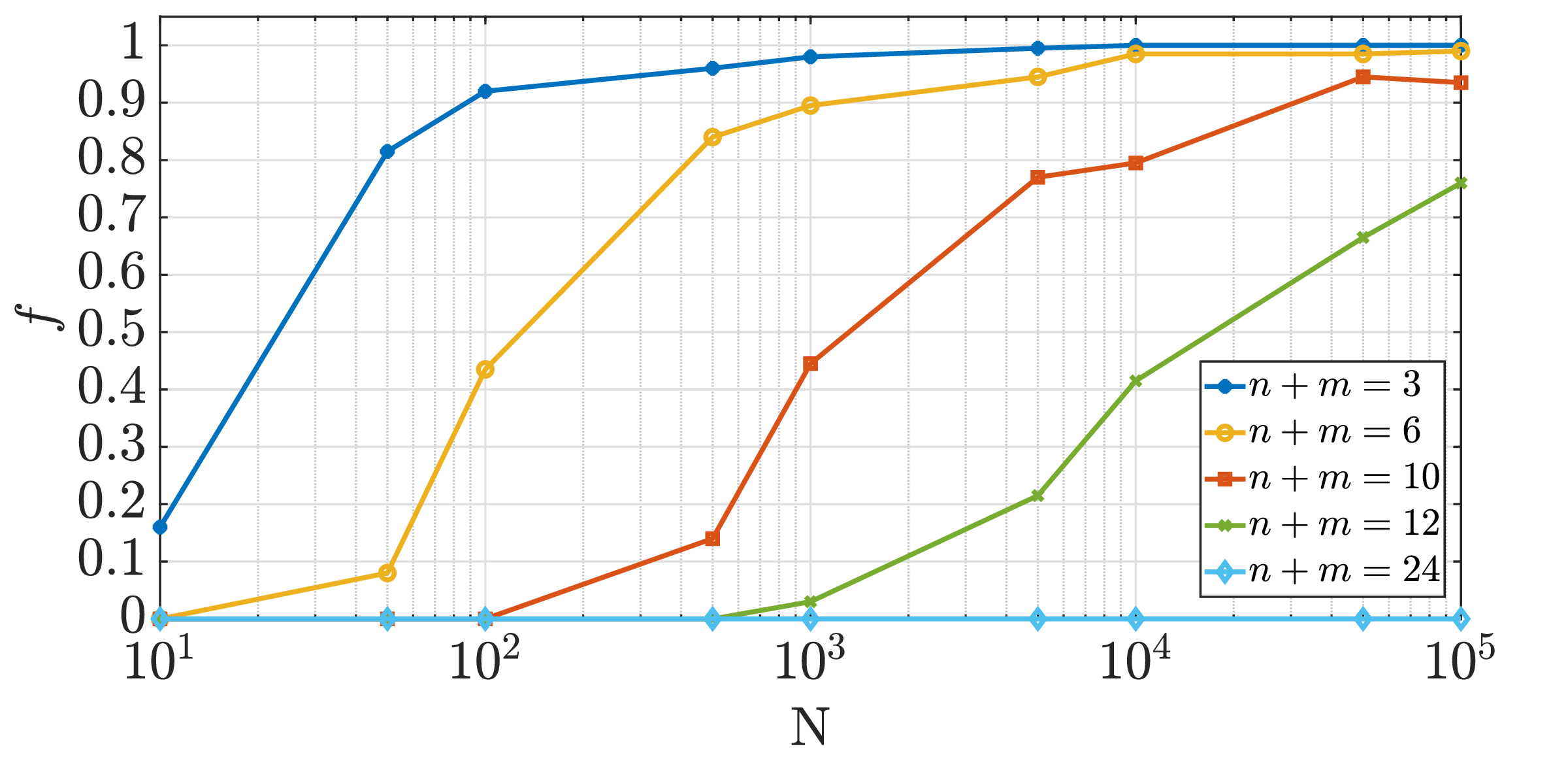}
\caption{Experiment 1: Relative frequency $f$ of Monte Carlo experiments leading to bounded LP problems as a function of the number of data points $N$ for different values of the state dimension $n$ and the number of control inputs $m$ of the system. }
\label{fig:bounded}
\end{figure}

\emph{Example 2:}
We perform a Monte Carlo simulation of $200$ experiments, where for each experiment:
\begin{enumerate} 
\item A system is generated as in Point 1 of Example 1.
\item
We apply Algorithm 1 with  $W=I$ in \eqref{eq:batch1} to determine the dataset $D$ and the covariance matrix $C;$ the dataset $D_1$ is collected by sampling $x_i \sim \Nc(0,1)$ and $u_i \sim \mathcal{U} (-2, 2).$ 
We solve the corresponding LP  \eqref{pb:LP_LQR_pi_data}.
\item
For the covariance matrix $C=I,$ we apply Algorithm 2 with $ \varepsilon_i = 0.01/ \Vert M_i \Vert$  for  $i=1,\dots, N$ in \eqref{eq:algo2_Ctilde} to select the dataset $D$  and the perturbed covariance $\tilde{C}$.
The dataset $D_1$ is collected by sampling $x_i \sim \Nc(0,1)$ and $u_i \sim \mathcal{U} (-2, 2)$.
We solve the LP  \eqref{pb:LP_LQR_pi_data} with the perturbed covariance $\tilde{C}.$ 
\item 
We randomly generate a dataset $D= \{x_i, u_i, l_i, x_i^+\}_{i=1}^{N}$  where $N = r + n+m,$ $x_i \sim \Nc(0,1)$ and $u_i \sim \mathcal{U} (-2, 2).$
We solve the convex approximation to the infinite LP \eqref{pb:LP_LQR_pi} proposed in \cite{regLP_CDC}
\beq  \label{eq:LP_pi_reg}
\begin{aligned} 
\tilde J_{\pi} = & \sup_{Q \in \mathbb{S}^{m+n}  }  \; \tr (QC)  \\
\text{s.t. }  &  \m{x_i \\ u_i}\tp Q \m{x_i \\ u_i} \leq l_i + \gamma  \m{x_i^+\\ K x_i^+ } \tp Q \m{x_i^+\\ K x_i^+ }, \\
&   \hspace{4.8cm} i=1,..,N, \\
&  \Vert Q \Vert\leq \theta .
\end{aligned} 
\eeq
Here, the last constraint acts as a regularization term; we choose the regularization parameter $\theta = 2 \Vert Q_{\pi} \Vert,$ where $Q_\pi$ is computed by solving the Lyapunov equation \ref{eq:lyap_Qpi}.
\item 
We assess the performance of the three algorithms by computing the error
$ e_{\pi} = \Vert \tilde{Q}_{\pi} - Q_{\pi} \Vert,  $
where $\tilde{Q}_{\pi}$ is the estimate of $Q_\pi$ obtained with one the previous approaches.
\end{enumerate}

Figure \ref{fig:algo3} displays the results of the simulations. 
The figure suggests that the strategies proposed in Section \ref{sec:smartsampl} are able to guarantee boundedness of the LP problem \eqref{pb:LP_LQR_pi_data} and to correctly recover the matrix $Q_\pi$ by using a small amount of data.
On the other hand, by randomly choosing a dataset with the same amount of data, the regularized approach does not lead to a good estimate of the q-function.
In most of the experiments, the solution to \eqref{eq:LP_pi_reg} is strongly influenced  by the regularization term, without which the problem would be unbounded.  
Figure \ref{fig:time} compares the computation time of Algorithm \ref{algo1} and Algorithm \ref{algo2} in one of the previous simulations for different values of the system dimensions $(m,n)$. As expected, Algorithm \ref{algo2}  is computationally more efficient.  

\begin{figure*}  
\centering
\begin{subfigure}{\columnwidth} 
\includegraphics[width=\linewidth]{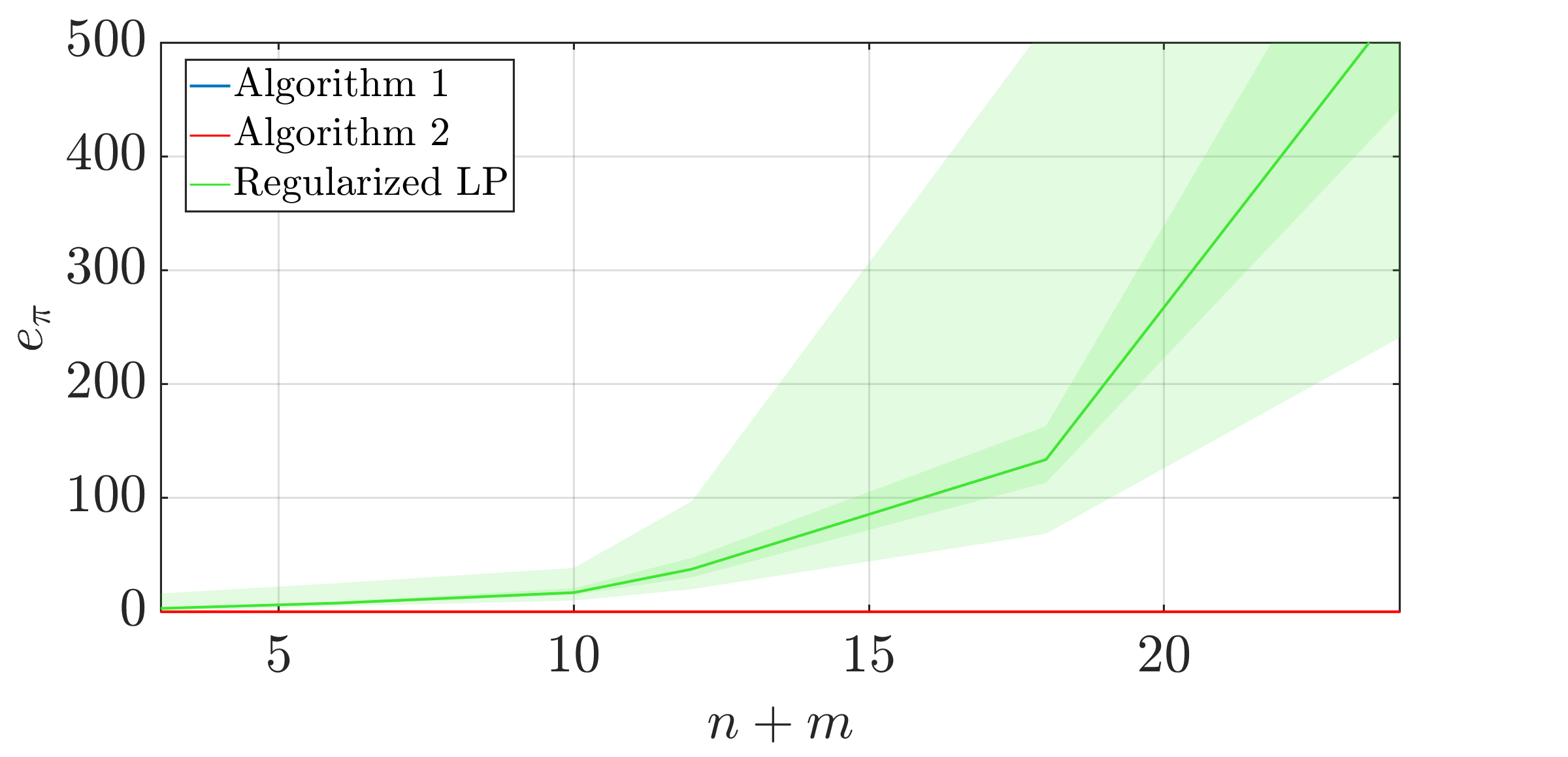}
\caption{}\label{fig:a}
\end{subfigure} 
\hfill
\begin{subfigure} {\columnwidth} 
\includegraphics[width=\linewidth]{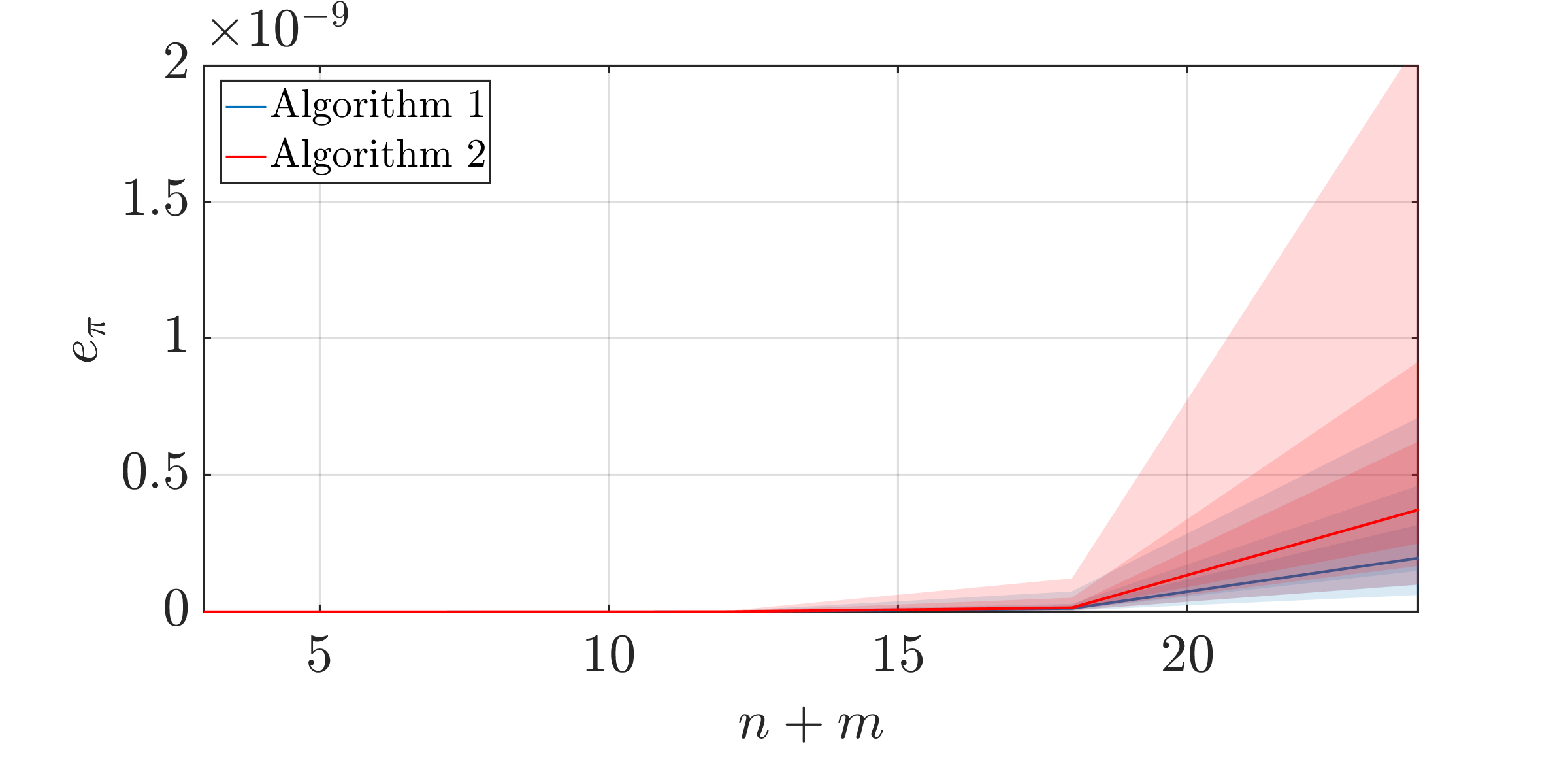}
\caption{}\label{fig:b}
\end{subfigure}
\caption{Experiment 2: Comparison of the error $e_\pi$ for the three different approaches for increasing values of the system dimensions $(n,m).$ 
The colored tubes represents the results between $[25\%,75\%]$ quantiles  and the median (solid lines) of $e_\pi$ across the 200 Monte Carlo experiments.  Figure \ref{fig:b} is a zoom of Figure \ref{fig:a}. }\label{fig:algo3}
\end{figure*}


\begin{figure}  
\includegraphics[width=\linewidth]{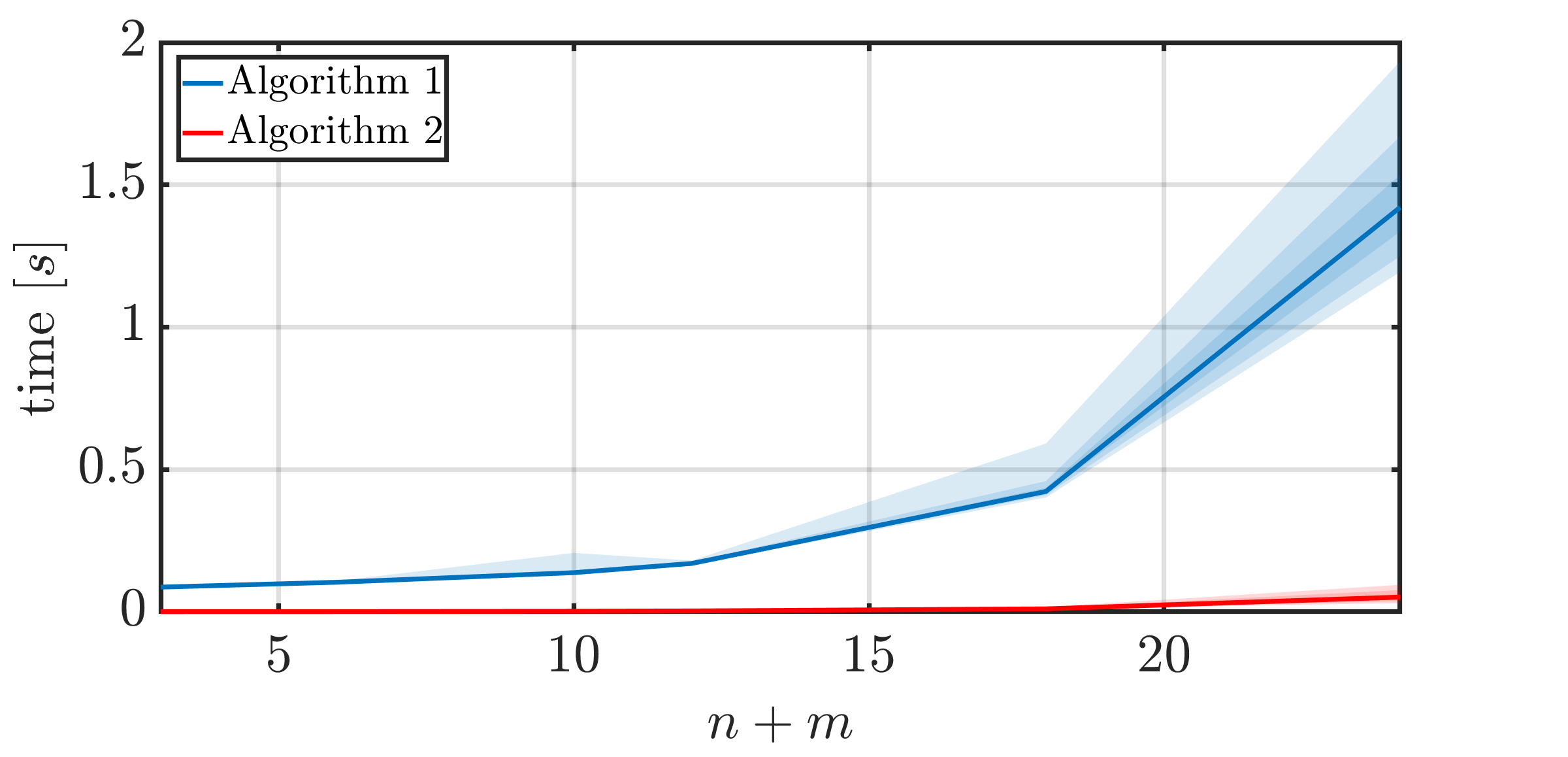}
\caption{Experiment 2: Comparison of the computation time of Algorithm \ref{algo1} and Algorithm \ref{algo2}   for increasing values of the system dimensions $(n,m).$  The colored tubes represents the results between $[25\%,75\%]$ quantiles  and the median (solid lines) of the computation time across the 200 Monte Carlo experiments. The simulations are run on a processor Intel\textregistered \hspace{0.01cm} 
 Core\texttrademark  \hspace{0.01cm} i5-8265U. } 
\label{fig:time}
\end{figure}


\section{Conclusions}

The LP approach is a promising method to tackle model-free infinite-dimensional optimal control problems. The main obstacle in its implementation is the difficulty of obtaining bounded solutions for reasonable amount of data. In this work we explore this question, navigating from exact to approximate, from infinite to finite programs, from LPs based on a policy to optimal LPs, and from nonlinear to linear systems. We discuss several fundamental properties on the geometry of the LPs, that culminate with necessary and sufficient condition for their boundedness. Extensive analysis is devoted to the data-driven LQR problem, where we show how to design the gradient in the objective function and construct a dataset to guarantee boundedness and optimality.  

A first research direction we envision is to specialize Theorem \ref{th:Farkas_pi} to a larger class of dynamical systems, \textit{e.g.}, polynomial, and give practical conditions on how to construct the dataset or the direction of growth to guarantee boundedness in the spirit of Algorithms \ref{algo1} and \ref{algo2}. Another challenging direction is to extend the present results to stochastic dynamical systems. In that case, one might start from the LP formulations introduced in \cite{cogill2006} and \cite{MARTINELLIAut}.

\setcounter{equation}{0}
\renewcommand{\theequation}{\thesection.\arabic{equation}}
\setcounter{lemma}{0}
\renewcommand{\thelemma}{\thesection.\arabic{lemma}}
\setcounter{theorem}{0}
\renewcommand{\thetheorem}{\thesection.\arabic{theorem}}
\setcounter{corollary}{0}
\renewcommand{\thecorollary}{\thesection.\arabic{corollary}}
\setcounter{propo}{0}
\renewcommand{\thepropo}{\thesection.\arabic{propo}}
\setcounter{assumption}{0}
\renewcommand{\theassumption}{\thesection.\arabic{assumption}}

\appendices
\section{Lyapunov Stability Theorems} 
We summarize some standard Lyapunov theorems for discrete-time LTI systems which will be extensively used in this paper. 
\begin{lemma} \label{lemm:lyap}
Consider the linear system 
$x_{t+1} = A x_t$ with $A \in \R^{n \times n}$. 
\begin{enumerate}
    \item If the spectral radius $\rho(A) < 1$ and $Z\in \Sbb^n,$ then 
    $ P = \sum_{k=0}^{\infty} (A\tp)^k Z A^k$  
    is the unique solution to the Lyapunov equation 
    $ P - A\tp P A = Z.$

    \item $\rho(A) < 1$  if and only if the Lyapunov LMI 
    \beq \label{app:LMI_lyap} 
    P - A\tp P A \succ 0,  \quad P \succ 0\eeq 
    is feasible.
    
    \item The maximum singular value $\sigma_{\max}(A) < 1$  if and only if $P=I$ is a feasible solution of the LMI \eqref{app:LMI_lyap}.
    \item  Let $\rho(A) < 1$    and $P\succ 0$ be any solution of the Lyapunov inequality \eqref{app:LMI_lyap}.
    Then  the matrix 
    $\bar A :=  P^{1/2} A  P^{-1/2} $ 
    is such that $     \sigma_{\max} (\bar A)  < 1 . $
 \end{enumerate}
\end{lemma} 
\begin{proof}
For the first two points we refer the reader to \cite[Ch.6.10.E]{antsaklis1997linear} 
and \cite[Ch.2]{boyd1994linear}, respectively.
Point 3 easily follows from the fact that $\sigma_{\max} (A) = \sqrt{\rho (AA\tp)}$ and 
$ AA\tp \preceq (\sigma_{\max}(A))^2 I .$ 
Point 4 is obtained by premultiplying and postmultiplying \eqref{app:LMI_lyap} by $P^{-1/2};$ the conclusion follows from Point 3. \end{proof}

\setcounter{equation}{0}
\renewcommand{\theequation}{\thesection.\arabic{equation}}
\setcounter{lemma}{0}
\renewcommand{\thelemma}{\thesection.\arabic{lemma}}
\setcounter{theorem}{0}
\renewcommand{\thetheorem}{\thesection.\arabic{theorem}}
\setcounter{corollary}{0}
\renewcommand{\thecorollary}{\thesection.\arabic{corollary}}
\setcounter{propo}{0}
\renewcommand{\thepropo}{\thesection.\arabic{propo}}
\setcounter{assumption}{0}
\renewcommand{\theassumption}{\thesection.\arabic{assumption}}

\section{Convex Cones} \label{app:convexcone}
We briefly recall some basic definitions and properties of convex cones. 
Let $V$ be a real vector space. 
A set $\Cc \subseteq V$ is called a \emph{cone} if for every $v \in \Cc$ and $\lambda \geq 0$ we have $\lambda v \in  \Cc.$ 
The set $\Cc$ is a \emph{convex cone} if it is convex and a cone, thus for every $v_1, v_2 \in \Cc$ and every $ \lambda_1, \lambda_2 \geq 0,$ the vector $ \lambda_1  v_1 + \lambda_2 v_2 \in \Cc. $ 
A convex cone $\Cc \subseteq V$ is called \emph{salient} if it contains no 1-dimensional vector subspace of $V$. 
A convex cone is salient if and only if $\Cc \cap (-\Cc) = \{0\}$ \cite[ch.2.6.1]{edwards2012functional}. 
A cone $\Cc$ is \emph{polyhedral} (or finitely-generated) if it is the conic combination of finitely many vectors, or, equivalently, if it is the intersection of a finite number of half-spaces. 
Let $\Cc$ be a convex  polyhedral cone. 
A linear inequality $\alpha\tp x \leq \beta$, with $\beta \in \R,$ is \emph{valid} for $\Cc$ if it satisfied for all points $x \in \Cc$.  
A \emph{face} of $\Cc$ is any set of the form $ F = \{x \in \Cc \; : \; \alpha\tp x = \beta \} $ where $\alpha\tp x \leq \beta  $ is some valid inequality for $\Cc.$ 
The faces of dimensions 0, 1, $\dim (\Cc) - 2$, and $\dim (\Cc) -1$ are
called \emph{vertices}, \emph{edges}, \emph{ridges}, and
\emph{facets}, respectively \cite{ziegler_polytopes}.

\setcounter{equation}{0}
\renewcommand{\theequation}{\thesection.\arabic{equation}}
\setcounter{lemma}{0}
\renewcommand{\thelemma}{\thesection.\arabic{lemma}}
\setcounter{theorem}{0}
\renewcommand{\thetheorem}{\thesection.\arabic{theorem}}
\setcounter{corollary}{0}
\renewcommand{\thecorollary}{\thesection.\arabic{corollary}}
\setcounter{propo}{0}
\renewcommand{\thepropo}{\thesection.\arabic{propo}}
\setcounter{assumption}{0}
\renewcommand{\theassumption}{\thesection.\arabic{assumption}}

\section{Proof of Theorem \ref{propo:LQR_opt_bounded}} \label{app:proof}
We start with a background lemma.
Recall that the pair $(A,B)$ is \emph{reachable} if and only if the controllability matrix $\m{B & AB & \dots & A^{n-1}B }$ is full rank \cite[ch.22]{MIT_lecturenotes}. 
\begin{lemma}\label{lemm:preliminary}
    Consider a reachable pair $(A,B)$     and matrices $R \in \Sbb^n$ and $P_{\max} \in \Sbb^n.$ 
   There exists a $P_{\min} \in \Sbb^n$ such that  all  $P \in \Sbb^n$ satisfying the LMIs
    \begin{align}  
        P B &= 0 \label{eq:lmi1}\\
        R + A\tp P A - P &\succeq 0  \label{eq:lmi2} \\
        P &\preceq P_{\max},  \label{eq:lmi3}
    \end{align}
   also satisy $P \succeq P_{\min}.$
\end{lemma}

\begin{proof}
 Consider a sequence $\{ P_k\}_{k=1}^{\infty}$ , with  $P_k \in \Sbb^n$ satisfying \eqref{eq:lmi1}-\eqref{eq:lmi3} for any $k,$ and let 
$$ P_k = \sum_{i=0}^{n} \lambda_{k,i} v_{k,i} v_{k,i}\tp$$
where $\{ v_{k,i} \}_{i=1}^{n} $ is an orthonormal basis of eigenvectors of $P_k$, and $\lambda_{k,i}$ are the corresponding eigenvalues in increasing order, 
$ \lambda_{k,1} \leq \lambda_{k,2} \leq ... \leq \lambda_{k,n}.$
Let $V_k = [ v_{k,1} \; \ldots \;  v_{k,n} ] \in \R^{n \times n} .$
The set $ \{ V \in \R^{n\times n} \; : \;  V = [ v_1 \; \dots \;  v_n ], \; \Vert v_i  \Vert = 1 \text{ for } i = 1, .., n, \; \text{ and }  v_i \tp v_j = 0, \text{ for } i\neq j, \; i,j = 1, .., n \} $ is compact,   
so the sequence $\{ V_{k} \}_{k=1}^{\infty}$ converges, up to a subsequence,  to a matrix 
$ \bar V = [ \bar{v}_1 \; \ldots \;  \bar{v}_n] \in 
\R^{n\times n}$ which is orthogonal.
Assume, for the sake of contradiction, that 
\beq \label{eq:abs_ass}
\lim_{k \to \infty} \lambda_{k,1} = -\infty.
\eeq
Clearly, as $ \Vert v_{k,1} \Vert = 1,$ 
\beq  \label{eq:norm1} 
\Vert\bar{v}_1 \Vert = 1 \eeq
and, by \eqref{eq:lmi1},
$ \bar{v}_1 \tp B = 0.$
Premultiplying and postmultiplying \eqref{eq:lmi2} by $B\tp$ and $B,$ respectively, 
and exploiting \eqref{eq:lmi1}, we obtain
\beq \label{eq:step1} 
B\tp A\tp P_k A B = \sum_{i=1}^{n} \lambda_{k,i} B\tp A\tp  v_{k,i} v_{k,i}\tp  A B
\succeq -B\tp R B  \eeq
for any $k \in \mathbb{N}.$ 
From \eqref{eq:lmi3}, $ \lambda_{k,i}< \infty$ for any $i$ and $k$, therefore \eqref{eq:abs_ass} and \eqref{eq:step1} leads to 
$  \bar{v}_1 \tp AB = 0.$
Similarly, by premultipying and postmultiplying \eqref{eq:lmi2} by $(AB)\tp$ and $A B$, respectively, we have:
$$  B\tp (A\tp)^2 P_k A^2 B  \succeq  B\tp A\tp P_k A B -B\tp A\tp R A B.$$
By, with an analogous reasoning  to the previous one, from  \eqref{eq:step1}we obtain that 
$ \bar{v}_1 \tp A^2 B = 0.$ 
Repeating this argument $n-1$ times, where at each step we premultiply and postmultiply the equation \eqref{eq:lmi2} by $(A^{t} B) \tp  $ and $A^{t} B,$ respectively, with $t=1,\dots,n-1$, we have that
$$\bar{v}_1 \tp \m{B & AB & \dots & A^{n-1}B } = 0.$$
Hence, since $(A,B)$ is reachable,
$\bar v_1 = 0.$
This contradicts equation \eqref{eq:norm1}, concluding the proof. 
\end{proof}

\paragraph*{\normalsize Proof of Theorem \ref{propo:LQR_opt_bounded}}
Assume for simplicity that $\gamma=1;$ note that under the conditions of the theorem this implies that $\rho(A)<1$, hence the open loop system is asymptotically stable. The proof can be extended to the case $\gamma \in (0,1)$ by replacing $(A,B)$ with $(\sqrt{\gamma}A,\sqrt{\gamma}B)$. 
Recall that  
$$L = \m{L_1 & L_{12} \\ L_{12}\tp & L_{2} }, \quad L_1 \in \mathbb{S}^{n}, L_2 \in \mathbb{S}^{m}.$$  
and partion $Q$ analogously.
We show necessity by contraposition:
assume that the system is not reachable and  show that the set $\Fc^*$ is unbounded. 
Let $Q_1 \in \Sbb^{n}$ be a non-zero matrix such that
\beq \label{eq:Q1} Q_1 = - \sum_{k=0}^{\infty} (A\tp)^k M A^k,\eeq
with $M=v_M v_M\tp \in \Sbb^n $ and $v_M \in \R^n$ non-zero belonging to the left-kernel of the reachability matrix, \emph{i.e.}

\beq \label{eq:ctrb_matrix} v_M\tp  \m{ B & AB & \dots & A^{n-1}B } = \m{ 0 & 0 & \dots & 0}.\eeq
Notice that $Q_1$ given by \eqref{eq:Q1} is well defined since $\rho(A) <1.$ Moreover, by point 1 of Lemma \ref{lemm:lyap}, $Q_1$ satisfies the Lyapunov equation 
\beq \label{eq:delta3} A\tp Q_1 A - Q_1  = M. \eeq
Finally, from \eqref{eq:Q1} and \eqref{eq:ctrb_matrix} it follows that
\beq \label{eq:delta1} 
Q_1 B = 0.
\eeq  
Consider now the matrix $Q = \m{Q_1 & 0 \\ 0 & 0}$. Substituting \eqref{eq:delta1} into \eqref{eq:LMI_Q} leads to
$$ \LMI(Q) = \m{L_1 & L_{12} & 0 \\ L_{12}\tp & L_2 & 0 \\ 0 & 0 & 0} + \m{ A\tp Q_1 A - Q_1 & 0 & 0\\ 0 & 0 & 0 \\ 0 & 0 & 0}.  $$
As $L \succeq 0,$ by \eqref{eq:delta3}, $\LMI(Q) \succeq 0,$ hence $Q \in \Fc^*.$
Since $v_M$ can be selected arbitrarily large,
we can build a matrix $Q$ belonging to the set $\Fc^*$  such that $\Vert Q \Vert \geq l $  for any $l \geq 0.$
This means that $\Fc^*$ is unbounded concluding the necessity proof.
For sufficiency, we show that, under the reachability assumption, any $Q\in \Fc^*$  is finite. 
Since the diagonal blocks of $\LMI(Q)$ are positive semidefinite, with standard algebraic manipulation, we obtain that if $ Q \in \Fc^*$ then
\begin{align}
    Q_1 &\preceq Q_{1,\max} := \sum_{k=0}^{\infty} (A\tp)^k L_1 A^k \label{eq:Q1_upp}\\
    0 \preceq Q_2 & \preceq Q_{2, \max} := L_2 + B\tp Q_{1,\max} B \label{eq:Q2_upp_low}. 
\end{align}
Moreover,
\begin{align} 
L_1 + A\tp Q_1 A - Q_1 & \succeq 0  \label{eq:Q1bound1} \\
B\tp Q_1 B + L_2 &\succeq 0. \label{eq:Q1bound2}
\end{align}
Let $\{ v_i \}_{i=1}^{n} $ be a set of orthonormal eigenvectors of $Q_1$ and  $\lambda_i \in \R$ the corresponding eigenvalues.  
Assume that $ v_i  \in \R^n \setminus \ker(B\tp)$ for $i=1, \dots , r$ and $ v_i  \in \ker(B\tp)$ for $i=r+1, \dots, n$ with $  0 \leq  r \leq n. $ 
Then, 
$$ Q_1 = \underbrace{\sum_{i=1}^{r} \lambda_i v_i v_i\tp}_{:=\tilde{Q}_1} + \underbrace{\sum_{i=r+1}^{n} \lambda_i v_i v_i\tp}_{:= \bar{Q}_1}. $$ 
Note that in case $r = 0$ then $\tilde Q_1 = 0,$ and when  $r=n$ then $ \bar Q_1 = 0.$  
From \eqref{eq:Q1bound2} it follows that 
$
B\tp \tilde{Q}_1 B  + L_2 \succeq 0, 
$
which, together with \eqref{eq:Q1_upp}, implies 
that  $\tilde{Q}_1 \succeq \tilde{Q}_{1,\min} $ for some $\tilde{Q}_{1,\min} \in \Sbb^{n}.$  
Moreover,  
$ \bar{Q}_1 B = 0$
and, from \eqref{eq:Q1bound1}, 
$ \bar{L}_1 + A\tp \bar{Q}_1 A - \bar{Q}_1 \succeq 0 $
with $ L_1 := L_1 + A\tp \tilde{Q}_1 A - \tilde{Q}_1 \in \Sbb^{n}.$
Therefore, from Lemma \ref{lemm:preliminary} there exists $\bar{Q}_{1,\min} \in \Sbb^{n}$  such that $\bar{Q}_1 \succeq \bar{Q}_{1,\min}.$  
We conclude that $Q_1 \succeq Q_{1,\min}$ where $Q_{1,\min}:= \tilde{Q}_{1,\min} + \bar{Q}_{1,\min}.$
Next, we show that $Q_{12}$ is bounded for any $Q \in \Fc^*.$ 
Indeed, from $\LMI (Q) \succeq 0 $ it follows that  
$$ \m{L_2 + B\tp Q_1 B - Q_2 &  B\tp Q_{12} \\ Q_{12}\tp B & Q_2 } \succeq 0.$$ 
Hence, by using the generalized Schur complement \cite[Theorem 1.20]{zhang2006schur}, 
\beq \label{eq:Schur} L_2 + B\tp Q_1 B - Q_2 - B\tp Q_{12} Q_2^\dagger Q_{12}\tp B \succeq 0, \eeq
where the symbol $^\dagger$ denotes the Moore–Penrose inverse.
Moreover, there exists $M \in \R^{m\times m},$  such that 
 $B\tp Q_{12} = M Q_2$ \cite[Theorem 1.19]{zhang2006schur}. 
Substituting into \eqref{eq:Schur} 
we obtain
$$ M Q_2 M\tp \preceq  L_2 + B\tp Q_1 B - Q_2 \preceq  L_2 + B\tp Q_{1, \max} B. $$
It follows that $\Vert MQ_2^{1/2} \Vert \leq \Vert (L_2 + B\tp Q_{1,\max} B)^{1/2} \Vert,$ hence $\Vert B\tp Q_{12} \Vert   \leq  \Vert (L_2 + B\tp Q_{1,\max} B)^{1/2} \Vert  \Vert Q_{2,\max}^{1/2} \Vert.$
Following the same steps for the principal submatrix of $\LMI(Q)$ obtained by removing the columns and the rows corresponding to the second block, we also have that 
$\Vert A\tp Q_{12} \Vert $ is bounded.
Since the system is reachable $\ker(B\tp) \cap \ker(A\tp) = \{0\}.$ 
Hence boundedness of $\Vert B\tp Q_{12}\Vert$ and $\Vert A\tp Q_{12}\Vert$ implies that $\Vert Q_{12}\Vert$ is itself bounded, concluding the proof.
\hfill $\QED$


\bibliographystyle{IEEEtran}      
\bibliography{biblio_LPforADP}


\begin{IEEEbiography}[{\includegraphics[width=1in,height=1.25in,clip,keepaspectratio]{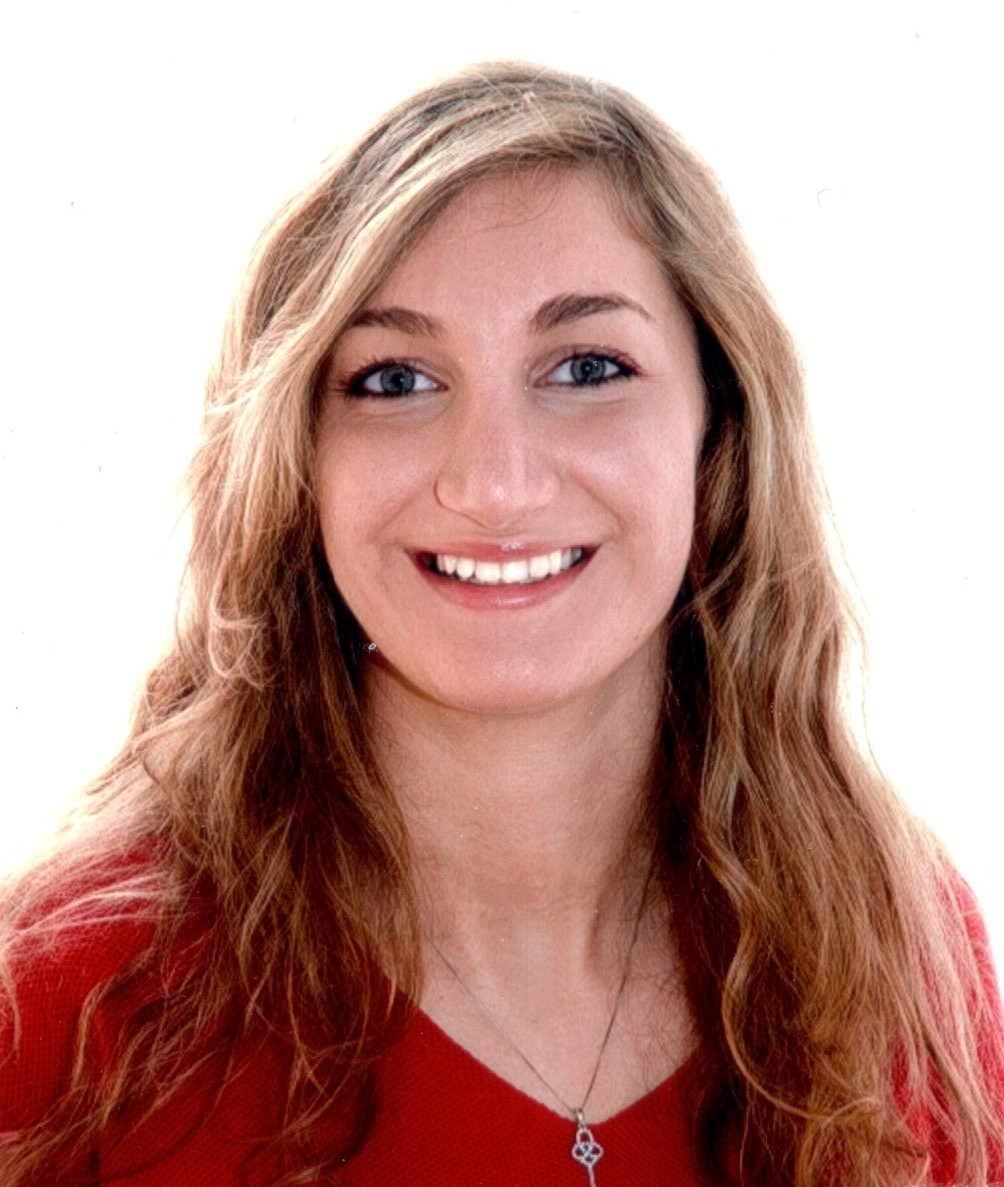}}]{Lucia Falconi}
received the B.Sc degree in Computer and Automation Engineering from Marche Polytechnic
University, Ancona, Italy, in 2018, and the M.Sc. degree in Automation
Engineering from the University
of Padova, Padova, Italy, in 2020. 
She is currently
pursuing the Ph.D. degree in Information Engineering
from the University of Padova. 
In 2022 she was a visiting Ph.D. student at the Automatic Control Laboratory, ETH Zürich, Switzerland.
Her research interests are in the areas of stochastic systems, systems identification and 
optimal control.
\end{IEEEbiography}

\begin{IEEEbiography}[{\includegraphics[trim={5cm 0 5cm 0},width=1in,height=1.25in,clip,keepaspectratio]{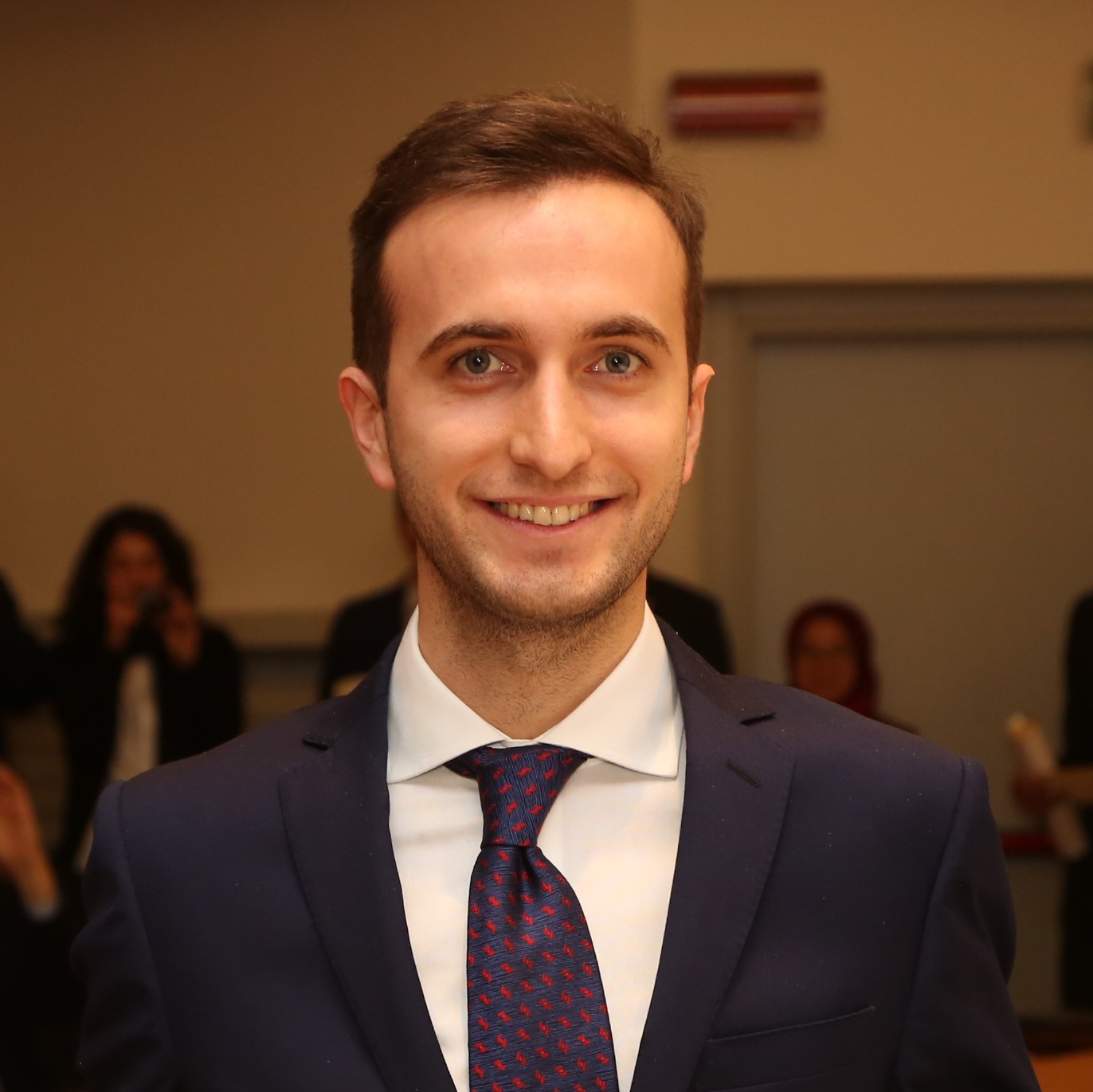}}]{Andrea Martinelli}
received the B.Sc. degree in management engineering and the M.Sc. degree in control engineering from Politecnico di Milano, Italy, in 2015 and 2017, respectively. He is currently working toward the Ph.D. degree at the Automatic Control Laboratory, ETH Zürich, Switzerland, since 2018. In 2017, he was a Master Thesis student at the Automatic Control Laboratory, EPF Lausanne, Switzerland. In 2018, he was a Research Assistant with the Systems \& Control Group at Politecnico di Milano. His research interests focus on data-driven optimal control and control of interconnected systems. 
\end{IEEEbiography}

\begin{IEEEbiography}[{\includegraphics[width=1in,height=1.25in,clip,keepaspectratio]{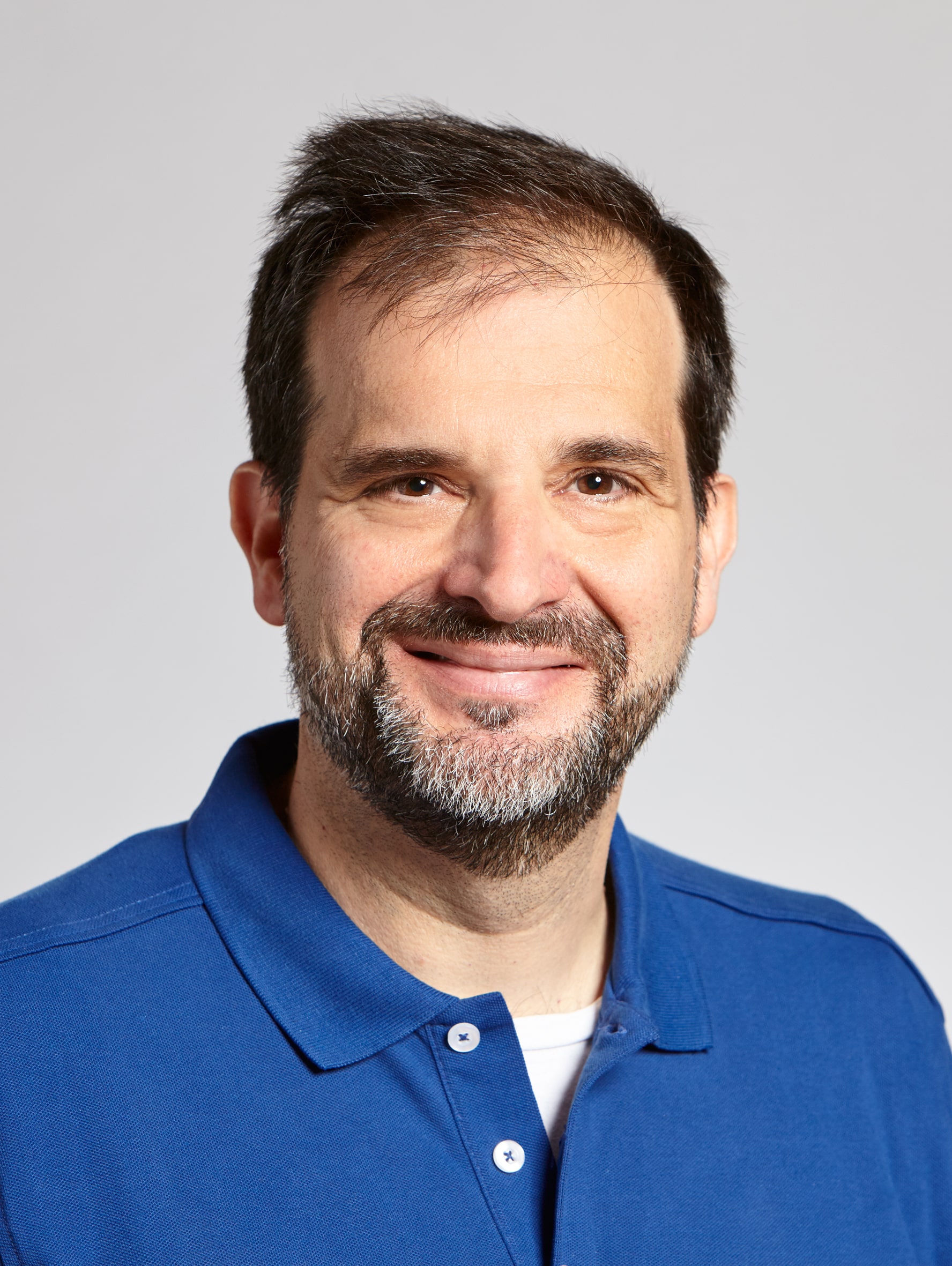}}]{John Lygeros} received a B.Eng. degree in 1990 and an M.Sc. degree in 1991 from Imperial College, London, U.K. and a Ph.D. degree in 1996 at the University of California, Berkeley. After research appointments at M.I.T., U.C. Berkeley and SRI International, he joined the University of Cambridge in 2000 as a University Lecturer. Between March 2003 and July 2006 he was an Assistant Professor at the Department of Electrical and Computer Engineering, University of Patras, Greece. In July 2006 he joined the Automatic Control Laboratory at ETH Zurich where he is currently serving as the Professor for Computation and Control and the Head of the laboratory. His research interests include modelling, analysis, and control of large-scale systems, with applications to biochemical networks, energy systems, transportation, and industrial processes. John Lygeros is a Fellow of IEEE, and a member of IET and the Technical Chamber of Greece. Since 2013 he is serving as the Vice-President Finances and a Council Member of the International Federation of Automatic Control and since 2020 as the Director of the National Center of Competence in Research "Dependable Ubiquitous Automation" (NCCR Automation).
\end{IEEEbiography}

\end{document}